\documentclass[lettersize,journal]{IEEEtran}
\usepackage{amsmath,amsfonts}
\usepackage{algorithm}
\usepackage{algpseudocode}
\usepackage{array}
\usepackage[caption=false,font=normalsize,labelfont=sf,textfont=sf]{subfig}
\usepackage{textcomp}
\usepackage{stfloats}
\usepackage{url}
\usepackage{verbatim}
\usepackage{amsthm}
\newtheorem{theorem}{Theorem}
\usepackage{graphicx}
\usepackage{cite}
\usepackage{pifont}
\usepackage{placeins}
\usepackage{caption}
\makeatletter
\renewcommand{\tagform@}[1]{\maketag@@@{\quad(#1)}}
\makeatother

\begin{document}

\title{CoMoE: Collaborative Optimization of Expert Aggregation and Offloading for MoE-based LLMs at Edge}

\author{
\IEEEauthorblockN{
Muqing Li\IEEEauthorrefmark{1},
Ning Li\IEEEauthorrefmark{1},~\IEEEmembership{Member,~IEEE},
Xin Yuan\IEEEauthorrefmark{1},
Wenchao Xu\IEEEauthorrefmark{2},
Quan Chen\IEEEauthorrefmark{3},
\\Song Guo\IEEEauthorrefmark{2},~\IEEEmembership{Fellow,~IEEE},
Haijun Zhang\IEEEauthorrefmark{4},~\IEEEmembership{Fellow,~IEEE}
}
\\
\IEEEauthorblockA{
\IEEEauthorrefmark{1}Harbin Institute of Technology 
\IEEEauthorrefmark{2}Hong Kong University of Science and Technology \\
\IEEEauthorrefmark{3}Guangdong University of Technology 
\IEEEauthorrefmark{4}University of Science and Technology Beijing
}

\IEEEauthorblockA{
\footnotesize 23s130418@stu.hit.edu.cn, li.ning@upm.es, xin.yuan@hit.edu.cn, \\
\hspace{1 cm} wenchaoxu@ust.hk, quanchen89@126.com, songguo@ust.hk, zhanghaijun@ustb.edu.cn
}
}

\markboth{Journal of \LaTeX\ Class Files,~Vol.~14, No.~8, August~2021}%
{Shell \MakeLowercase{\textit{et al.}}: A Sample Article Using IEEEtran.cls for IEEE Journals}

\maketitle

\begin{abstract}
The proliferation of large language models (LLMs) has driven the adoption of Mixture-of-Experts (MoE) architectures as a promising solution to scale model capacity while controlling computational costs. However, deploying MoE models in resource-constrained mobile edge computing environments presents significant challenges due to their large memory footprint and dynamic expert activation patterns. To address these challenges, we propose a novel dynamic resource-aware collaborative optimization framework that jointly optimizes expert aggregation granularity and offloading strategies based on real-time device resource states, network conditions, and input characteristics in mobile edge environments, denoted as CoMoE. In CoMoE, we first systematically analyze existing expert aggregation techniques, including expert parameter merging, knowledge distillation, and parameter sharing decomposition, identifying their limitations in dynamic mobile environments. We then investigate expert offloading strategies encompassing expert prediction and prefetching, expert caching and scheduling, and multi-tier storage architectures, revealing the interdependencies between routing decisions and offloading performance. The CoMoE incorporates adaptive scheduling mechanisms that respond to user mobility and varying network conditions, enabling efficient MoE deployment across heterogeneous edge devices. Extensive experiments on real mobile edge testbeds demonstrate that CoMoE achieves approximately 70\% reduction in memory usage compared to baseline methods, 10.5\% lower inference latency than existing expert offloading techniques, while maintaining model performance stability. For large-scale MoE models (e.g., 7.4B-parameter Switch-Base-128), the CoMoE reduces memory requirements from 15.6GB to 4.7GB, enabling deployment on resource-constrained mobile edge devices that previously could only support much smaller models.
\end{abstract}

\begin{IEEEkeywords}
Edge intelligence, Large language models, Mixture of Experts, Expert aggregation, Expert offloading
\end{IEEEkeywords}

\section{Introduction}
With the rapid advancement of artificial intelligence technology, Large Language Models (LLMs) have demonstrated unprecedented capabilities in natural language processing, computer vision, and other domains. However, as model scales continue to expand, computational efficiency and memory constraints have become critical challenges in practical model deployment. The Mixture of Experts (MoE) architecture emerges as a promising solution that effectively scales the model capacity while controlling computational costs through sparse activation mechanisms. Unlike traditional dense models, MoE models selectively activate only a subset of experts for each input token, thereby achieving significant computational savings during inference.

Despite the remarkable success of MoE models in cloud environments, their deployment in edge network environments still faces multiple fundamental challenges. First, there exists a severe mismatch between the high resource demands of the models and the limited resources available in edge networks. Although MoE models improve computational efficiency through sparse activation, they require to store all expert parameters, which makes their total parameter count typically large. This characteristic of ``storage redundancy, but computational efficiency'' becomes a major bottleneck under the limited storage resources of edge environments. Second, expert scheduling conflicts with communication overhead. The inference process of MoE models involves the selection and activation of dynamic experts. In edge environments, experts typically need to be distributed across different storage hierarchies or even different devices, leading to frequent loading and unloading of expert parameters during inference. This introduces significant communication latency in bandwidth-constrained edge networks, where expert loading and communication overhead can contribute to 40\%-60\% of total inference latency. Furthermore, the high dynamism and heterogeneity of edge network environments further exacerbate deployment difficulties. The computational capacity, storage capacity, and network bandwidth of devices are not only limited but also exhibit significant volatility, requiring MoE model deployment schemes to adapt to resource changes in real-time.

To address these challenges, existing research primarily explores two directions: expert aggregation and expert offloading. Expert aggregation techniques reduce computational or storage overhead by combining multiple expert parameters or outputs, including expert parameter merging (such as MEO~\cite{meo}), knowledge distillation (such as EEP~\cite{eep}), and parameter sharing and decomposition (such as WideNet~\cite{widenet}, MPoE~\cite{mpoe}). Expert offloading strategies address memory bottlenecks by transferring inactive experts from GPU to host memory or external storage, including expert prediction and prefetching based on inter-layer correlations (such as MixtralOffloading~\cite{mixtral}, Pre-gated MoE~\cite{pregated}, EdgeMoE~\cite{edgemoe}), methods based on historical inference data (such as MoE-Infinity~\cite{moeinfinity}, ProMoE~\cite{promoe}), model-based prediction methods (such as DyNN-Offload~\cite{dynnoffload}, SiDA~\cite{sida}), and expert caching and scheduling strategies based on access frequency and importance (such as SwapMoE~\cite{swapmoe}, AdapMoE~\cite{adapmoe}, HOBBIT~\cite{hobbit}, CacheMoE~\cite{cachemoe}). These strategies successfully improve the adaptability of large MoE models to resource-constrained edge environments, significantly improving their performance in edge networks. However, despite these advances, existing solutions exhibit several significant limitations: 1) The Limitations of Single Optimization: Most current research tends to consider expert aggregation or offloading strategies separately, lacking a systematic study of their collaborative optimization. This kind of strategy often leads to suboptimal overall performance. 2) The effectiveness of Adaptability on dynamic Strategies: The majority of existing methods are designed for relatively stable resource environments and generally fail to account for the dynamic heterogeneity of edge network resource states and varying application requirements. This disadvantage affects their effectiveness in dynamic and heterogeneity edge deployments. 3) The tight coupling of Routing-Offloading Interdependencies: the existing strategies fail to analyze and optimize the interdependent relationships between routing decisions (which experts to activate) and offloading strategies (where to store experts) in depth. However, overlooking this interdependency can lead to inefficiencies and serious performance degradation.

For addressing the above issues, this paper proposes a novel collaborative optimization framework, which is designed to enhance the efficiency of Mixture-of-Experts (MoE) model inference in resource-constrained and dynamic edge environments, denoted CoMoE. The main advantage of this approach can be summarized: this approach integrates dynamic resource perception, adaptive expert aggregation, and intelligent expert offloading strategies, and addresses the inherent challenges of resource dynamic, heterogeneity, and the mismatch between MoE model demands and edge device capabilities. Specifically, we formalize the problem as a multi-objective optimization, decompose it into manageable subproblems, and propose a dynamic resource perception module for real-time monitoring and prediction. Furthermore, we develop a resource-aware expert aggregation method, including fixed-ratio and adaptive-ratio fusion, and an expert offloading strategy with multilevel storage and activation prediction. This comprehensive strategy, with its two-level resource adaptation mechanism, aims to balance computational performance and resource adaptability, particularly optimizing for Encoder-Decoder architectures.

\subsection{Contributions}

To address the aforementioned limitations and advance the state of the art in MoE model deployment on edge networks, this paper makes the following \textit{novel and significant} contributions.

\begin{itemize}
\item \textbf{First Unified Collaborative Optimization Framework for Expert Aggregation and Offloading}: Unlike existing approaches that treat expert aggregation and offloading as separate optimization problems, we pioneer the first systematic framework that jointly optimizes both strategies. Our key innovation lies in identifying and mathematically modeling the interdependencies between routing decisions, expert aggregation granularity, and offloading performance. This collaborative approach achieves up to 25\% better resource utilization compared to independent optimization strategies, fundamentally shifting the paradigm from isolated optimization to holistic system level efficiency.

\item \textbf{Novel Dynamic Resource Aware Adaptation Mechanism with Two-Level Granularity}: We introduce an innovative dual-granularity adaptation mechanism that responds to edge resource dynamics on both coarse-grained (seconds to minutes) and fine-grained (milliseconds) time scales. The coarse grained level employs precomputed model variants with different aggregation ratios, while the fine grained level performs real-time expert scheduling. This hierarchical design is the first to address resource volatility across multiple temporal dimensions, enabling 40-60\% reduction in expert loading overhead under dynamic conditions compared to static approaches.

\item \textbf{Architecture Specific Encoder Decoder Optimization Strategy}: We propose the first specialized optimization approach that differentiates between Encoder and Decoder components in MoE architectures. Our key insight is that Encoders benefit from aggressive aggregation and offloading due to their parallel processing nature, while Decoders require minimal intervention due to their sequential auto-regressive generation pattern. This architecture-aware strategy achieves a 35\% improvement in memory trade-offs compared to uniform optimization approaches.

\end{itemize}

\textbf{Collective Impact}: These innovations collectively enable, for the first time, the deployment of large-scale MoE models (up to 15.8B parameters) on resource-constrained edge devices with memory footprints reduced by 70\% while maintaining near-optimal performance. Our comprehensive solution fundamentally transforms the feasibility landscape of edge-based large language model deployment, opening new possibilities for democratized AI applications in mobile and IoT environments.

\section{Problem Formulation}

This section formally defines the system model, analyzes the inherent challenges, and presents the mathematical formulation of the optimization problem addressed in this work.

\subsubsection{System Model}

We consider a distributed inference scenario where Mixture-of-Experts (MoE) based large language models are deployed across a collection of edge devices. Let $\mathcal{D} = \{d_1, d_2, \ldots, d_n\}$ denote the set of $n$ edge devices. Each device $d_i$ is characterized by dynamic and heterogeneous resource constraints that are critical for efficient MoE model operation:

\begin{itemize}
\item \textbf{Computational Resources:} $C_i(t) = \{c_i^{\text{gpu}}(t), c_i^{\text{cpu}}(t)\}$ represents the computational capabilities available (e.g., floating-point operations per second, FLOPS) of the GPU and CPU on device $d_i$ at time $t$.

\item \textbf{Memory Resources:} $M_i(t) = \{m_i^{\text{gpu}}(t), m_i^{\text{cpu}}(t)\}$ denotes the available memory capacity (e.g., in gigabytes) of the GPU and CPU on the device $d_i$ at time $t$.

\item \textbf{Bandwidth Resources:} $B_i(t) = \{b_i^{\text{gpu-cpu}}(t), b_i^{\text{net}}(t)\}$ represents the internal bandwidth between the GPU and the CPU (for example, PCIe bandwidth) and the network bandwidth (for example, Wi-Fi or cellular) of the device $d_i$ at time $t$.
\end{itemize}

The MoE model, denoted as $\mathcal{M}$, consists of $L$ MoE layers, each containing $E$ distinct experts. The complete set of experts within the model is defined as $\mathcal{E}_{\mathcal{M}} = \{e_{l,j} \mid l \in [1, L], j \in [1, E]\}$, where $e_{l,j}$ refers to the $j$-th expert in the $l$-th layer. Each expert $e_{l,j}$ has a specific parameter size, denoted $s_{l,j}$.

During inference, for a given input sequence $X = \{x_1, x_2, \ldots, x_T\}$, each input token $x_t$ traverses the MoE layers. In the $l$-th layer, a router mechanism selects a subset of experts $A(x_t, l) = \{e_{l,j_1}, e_{l,j_2}, \ldots, e_{l,j_K}\}$ for processing. Here, $K$ represents the number of experts activated, which is typically much smaller than the total number of experts in the layer ($K \ll E$).

\subsubsection{Challenges in Edge MoE Deployment}

Deploying large-scale MoE models in dynamic edge environments presents several significant challenges that necessitate a robust optimization framework.

\begin{itemize}
\item \textbf{Resource Mismatch:} The cumulative parameter size of MoE models, represented by $\sum_{l=1}^{L} \sum_{j=1}^{E} s_{l,j}$, often substantially exceeds the available GPU memory ($m_i^{\text{gpu}}(t)$) on typical edge devices. This fundamental mismatch prevents the direct deployment of the complete model on a single-edge device.

\item \textbf{Resource Dynamism:} The resources of the edge device are inherently dynamic, with resource states fluctuating over time. Mathematically, this implies that the partial derivatives of computational, memory, and bandwidth resources with respect to time are non-zero (i.e., $\frac{\partial C_i(t)}{\partial t} \neq 0$, $\frac{\partial M_i(t)}{\partial t} \neq 0$, $\frac{\partial B_i(t)}{\partial t} \neq 0$). Such volatility renders static deployment strategies ineffective and necessitates adaptive solutions.

\item \textbf{Resource Heterogeneity:} Edge environments exhibit significant variation in resource configurations among different devices. This heterogeneity, expressed as $C_i(t) \neq C_j(t)$, $M_i(t) \neq M_j(t)$, and $B_i(t) \neq B_j(t)$ for any $i \neq j$, complicates the design of a unified and universally optimal deployment strategy.

\item \textbf{Expert Loading Overhead:} Frequent transfer of expert parameters between different memory hierarchies (e.g., from CPU memory to GPU memory) introduces considerable communication overhead. This overhead is particularly pronounced when an expert $e_{l,j}$ is not resident in GPU memory ($e_{l,j} \notin \mathcal{M}^{\text{gpu}}$), requiring additional loading time that can significantly impact inference latency.
\end{itemize}

\subsubsection{Problem Formalization}

Given the system model and challenges described above, we formalize the efficient deployment of MoE models in edge environments as a multi-objective optimization problem. The primary objective is to minimize end-to-end inference latency while strictly adhering to resource constraints and ensuring that model performance remains above a predefined threshold.

The optimization problem can be formally expressed as
\begin{align}
\min_{\mathcal{P},\mathcal{S}} &\quad L_{\text{total}}(X, \mathcal{M}, \mathcal{P}, \mathcal{S}, \mathcal{D}) \label{eq:optimization_problem}\\
\text{s.t.} &\quad P(\mathcal{M}, \mathcal{P}) \geq P_{\text{threshold}} \label{eq:performance_constraint}\\
&\quad M_{\text{usage}}(\mathcal{M}, \mathcal{P}, \mathcal{S}, d_i) \leq M_i(t), \quad \forall d_i \in \mathcal{D}, \forall t \label{eq:memory_constraint}
\end{align}

\noindent Where:
\begin{itemize}
\item $\mathcal{P}$ denotes the expert aggregation strategy, which dictates how experts are grouped and merged to reduce the footprint of the model.
\item $\mathcal{S}$ represents the expert offloading strategy, governing the distribution and dynamic scheduling of experts across various storage hierarchies (e.g., GPU memory, CPU memory).
\item $L_{\text{total}}$ represents the total end-to-end inference latency, covering both computational and communication delays.
\item $P(\mathcal{M}, \mathcal{P})$ is the performance metric of the aggregated model, such as accuracy or F1 score.
\item $P_{\text{threshold}}$ defines the minimum acceptable performance level for the model.
\item $M_{\text{usage}}$ quantifies the memory consumption of the MoE model, which must not exceed the available memory $M_i(t)$ on any device $d_i$ at any given time $t$.
\end{itemize}

The total end-to-end inference latency, $L_{\text{total}}$, can be further decomposed into computational and communication components:
\begin{align}
L_{\text{total}}(X, \mathcal{M}, \mathcal{P}, \mathcal{S}, \mathcal{D}) &= \sum_{t=1}^{T} \sum_{l=1}^{L} \Big[L_{\text{comp}}(x_t, l, \mathcal{M}, \mathcal{P}) \nonumber \\
&\quad + L_{\text{comm}}(x_t, l, \mathcal{M}, \mathcal{P}, \mathcal{S}, \mathcal{D})\Big] \label{eq:total_latency}
\end{align}

Here, $L_{\text{comp}}$ represents the computation latency incurred during the forward pass of an MoE layer, and $L_{\text{comm}}$ denotes the communication latency primarily attributed to expert parameter loading.

The communication latency $L_{\text{comm}}$ for a given token $x_t$ at layer $l$ can be specifically modeled as:
\begin{multline}
L_{\text{comm}}(x_t, l, \mathcal{M}, \mathcal{P}, \mathcal{S}, \mathcal{D}) = \sum_{e_{l,j} \in A(x_t, l)} \mathbb{I}(e_{l,j} \notin \mathcal{M}^{\text{gpu}}) \\
\cdot \frac{s_{l,j}}{b_i^{\text{gpu-cpu}}(t)} 
\label{eq:comm_latency}
\end{multline}

where $\mathbb{I}(e_{l,j} \notin \mathcal{M}^{\text{gpu}})$ is an indicator function that evaluates to 1 if expert $e_{l,j}$ is not present in the GPU memory and 0 otherwise. The term $s_{l,j} / b_i^{\text{gpu-cpu}}(t)$ represents the time required to transfer the expert parameters $e_{l,j}$ from the CPU to GPU memory.

Given the NP-hard nature of this multi-objective optimization problem, direct analytical solutions are impractical. We therefore propose a decomposition approach that breaks down the problem into three interconnected subproblems:

\begin{enumerate}
\item \textbf{Resource-Awareness Problem:} How can we accurately perceive and predict the dynamic and heterogeneous resource states of edge devices?
\item \textbf{Expert Aggregation Problem:} How can we dynamically adjust the granularity of expert aggregation based on perceived resource status to effectively reduce the model's parameter footprint while preserving performance?
\item \textbf{Expert Offloading Problem:} How can we optimize the distribution and scheduling of experts across different storage hierarchies by leveraging predicted expert activation patterns and current resource states?
\end{enumerate}

\section{CoMoE}

This section details the proposed collaborative optimization framework for expert aggregation and offloading, specifically designed for resource-constrained edge environments. The framework dynamically adjusts the granularity of expert aggregation and offloading strategies by comprehensively considering the heterogeneous and time-varying resource constraints of edge devices. This adaptive approach aims to achieve efficient inference for Mixture-of-Experts (MoE) models. The overall architecture of the proposed framework is illustrated in Figure~\ref{fig:overall_architecture}.

\begin{figure*}[!t]
\centering
\includegraphics[width=\textwidth]{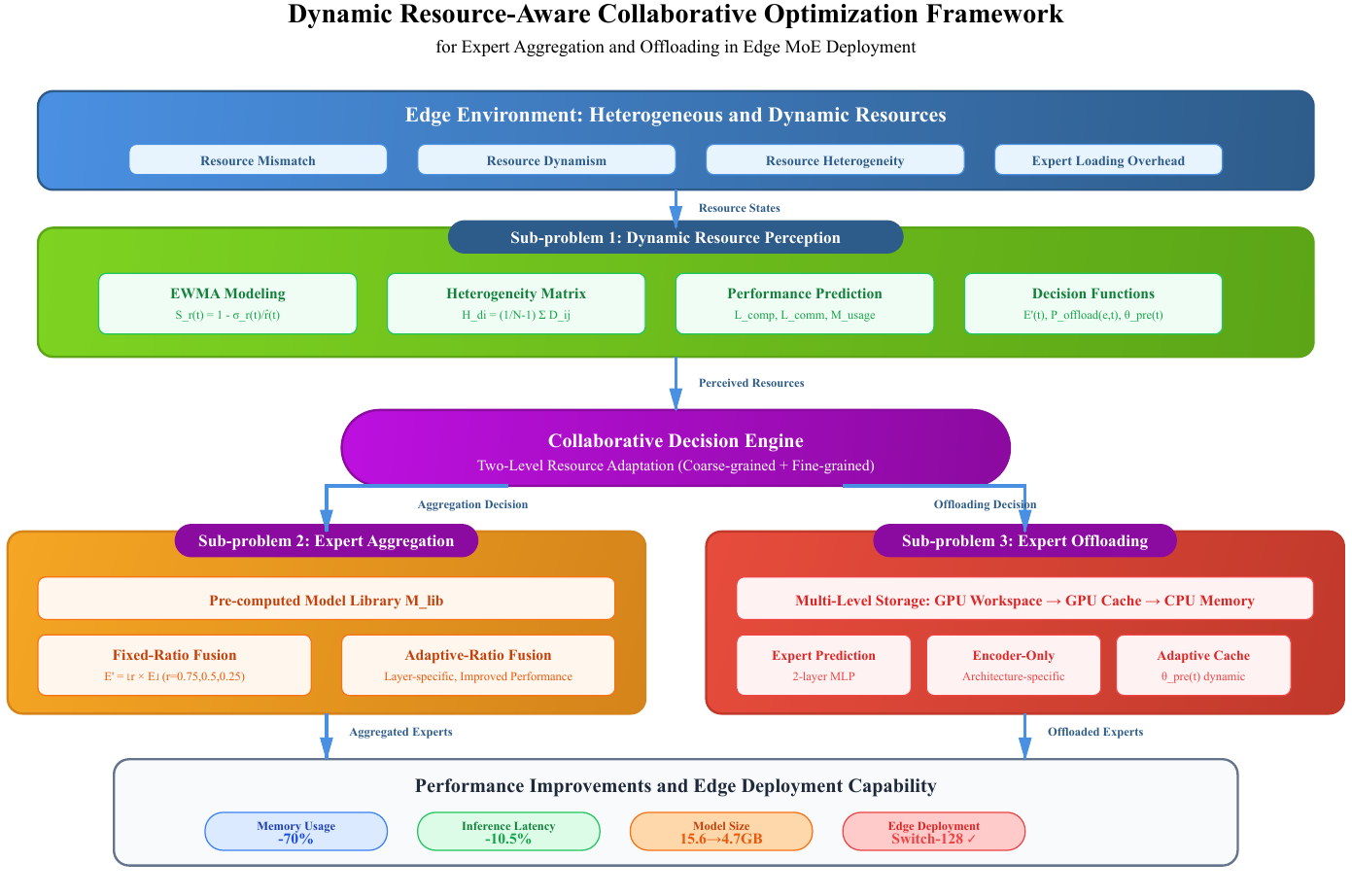}
\caption{Overall Architecture of the Proposed Collaborative Optimization Framework}
\label{fig:overall_architecture}
\end{figure*}

\subsection{Dynamic Heterogeneous Resource Perception and Modeling}

To effectively manage the highly dynamic and heterogeneous nature of resources in edge networks, we introduce a dynamic resource perception module. This module is designed for realtime monitoring and modeling of the computational capabilities, memory capacity, and network bandwidth of edge devices.

\subsubsection{Edge Resource State Characterization}

The resource state of an edge device $d_i$ at time $t$ is represented as a multidimensional vector $R_i(t)$, which inherently changes over time:
\begin{equation}
R_i(t) = \{C_i(t), M_i(t), B_i(t)\}
\end{equation}

where $C_i(t)$ denotes the computational resource state, $M_i(t)$ represents the memory resource state, and $B_i(t)$ represents the bandwidth resource state. These components are further defined as follows:

\begin{itemize}
\item \textbf{Computational Resource State $C_i(t)$:}
\begin{equation}
C_i(t) = \{c_i^{\text{gpu}}(t), c_i^{\text{cpu}}(t), u_i^{\text{gpu}}(t), u_i^{\text{cpu}}(t)\}
\end{equation}
Here, $c_i^{\text{gpu}}(t)$ and $c_i^{\text{cpu}}(t)$ represent the raw computational capabilities (e.g., in FLOPS) of the GPU and CPU, respectively. $u_i^{\text{gpu}}(t)$ and $u_i^{\text{cpu}}(t)$ indicate their current utilization rates.

\item \textbf{Memory Resource State $M_i(t)$:}
\begin{equation}
M_i(t) = \{m_i^{\text{gpu}}(t), m_i^{\text{cpu}}(t), m_i^{\text{used,gpu}}(t), m_i^{\text{used,cpu}}(t)\}
\end{equation}
Here, $m_i^{\text{gpu}}(t)$ and $m_i^{\text{cpu}}(t)$ denote the total memory capacity of the GPU and CPU, respectively. $m_i^{\text{used,gpu}}(t)$ and $m_i^{\text{used,cpu}}(t)$ represent the amount of memory currently in use.

\item \textbf{Bandwidth Resource State $B_i(t)$:}
\begin{equation}
B_i(t) = \{b_i^{\text{gpu-cpu}}(t), b_i^{\text{net}}(t), l_i^{\text{gpu-cpu}}(t), l_i^{\text{net}}(t)\}
\end{equation}
Here, $b_i^{\text{gpu-cpu}}(t)$ refers to the bandwidth between the GPU and CPU, while $b_i^{\text{net}}(t)$ represents the network bandwidth. Correspondingly, $l_i^{\text{gpu-cpu}}(t)$ and $l_i^{\text{net}}(t)$ denote the associated latencies.
\end{itemize}

\subsubsection{Resource Dynamism Modeling}

\begin{algorithm}[!t]
\caption{Resource Dynamism Modeling}
\label{alg:resource_dynamism_modeling}
\begin{algorithmic}[1]
\Require Resource indicator $r(t)$, smoothing coefficient $\alpha$, time window $w$
\Ensure Smoothed resource indicator $\hat{r}(t)$, resource stability indicator $S_r(t)$
\State Initialize $\hat{r}(0) = r(0)$
\For{$t = 1$ to end of time series}
    \State $\hat{r}(t) = \alpha \cdot r(t) + (1 - \alpha) \cdot \hat{r}(t-1)$
    \State $\sigma_r(t) = \sqrt{\frac{1}{w} \sum_{i=t-w+1}^{t} (r(i) - \hat{r}(t))^2}$
    \State $S_r(t) = 1 - \frac{\sigma_r(t)}{\hat{r}(t)}$
\EndFor
\State \Return $\hat{r}(t), S_r(t)$
\end{algorithmic}
\end{algorithm}

To accurately capture the dynamic fluctuations inherent in edge resources, we employ time series analysis to model resource state variations. For any given resource indicator $r(t)$, we apply the Exponentially Weighted Moving Average (EWMA) method to smooth out short-term fluctuations and derive a more stable estimate:
\begin{equation}
\hat{r}(t) = \alpha \cdot r(t) + (1 - \alpha) \cdot \hat{r}(t-1)
\end{equation}

where $\hat{r}(t)$ is the smoothed resource indicator value, and $\alpha \in (0, 1)$ is the smoothing coefficient. Building upon this, we define the resource stability indicator $S_r(t)$ as:
\begin{equation}
S_r(t) = 1 - \frac{\sigma_r(t)}{\hat{r}(t)}
\end{equation}

where $\sigma_r(t)$ represents the standard deviation of the resource indicator $r$ within the most recent $w$ time windows, calculated as:
\begin{equation}
\sigma_r(t) = \sqrt{\frac{1}{w} \sum_{i=t-w+1}^{t} (r(i) - \hat{r}(t))^2}
\end{equation}

A higher value of $S_r(t)$ indicates greater resource stability, whereas a lower value suggests more significant resource fluctuations.

\subsubsection{Resource Heterogeneity Quantification}

\begin{algorithm}[!t]
\caption{Resource Heterogeneity Quantification}
\label{alg:resource_heterogeneity_quantification}
\begin{algorithmic}[1]
\Require Set of edge devices $D = \{d_1, \dots, d_N\}$, resource dimensions $R$, weights $w_r$
\Ensure Resource difference matrix $\mathbf{D}$, heterogeneity scores $H_{d_i}$
\State Initialize $N \times N$ matrix $\mathbf{D}$ with zeros
\For{each device $d_i \in D$}
    \For{each device $d_j \in D$ where $j \neq i$}
        \State Calculate resource difference $D_{ij} = \sum_{r \in R} w_r \cdot \frac{|r_{d_i} - r_{d_j}|}{\max(r_{d_i}, r_{d_j})}$
    \EndFor
\EndFor
\For{each device $d_i \in D$}
    \State Calculate heterogeneity score $H_{d_i} = \frac{1}{N-1} \sum_{j \neq i} D_{ij}$
\EndFor
\State \Return $\mathbf{D}, H_{d_i}$
\end{algorithmic}
\end{algorithm}

To quantify the degree of resource heterogeneity among different edge devices, we introduce a resource difference matrix $\mathbf{D}$. For any two devices $d_i$ and $d_j$, their resource difference is defined as:
\begin{equation}
D_{ij} = \sum_{r \in R} w_r \cdot \frac{|r_{d_i} - r_{d_j}|}{\max(r_{d_i}, r_{d_j})}
\end{equation}

where $r$ iterates over all defined resource dimensions, and $w_r$ is the weight assigned to the corresponding resource dimension, satisfying the condition $\sum_r w_r = 1$. Based on this resource difference matrix, we can further define the heterogeneity score $H_{d_i}$ for device $d_i$ as:
\begin{equation}
H_{d_i} = \frac{1}{N-1} \sum_{j \neq i} D_{ij}
\end{equation}

where $N$ is the total number of devices in the system. A higher heterogeneity score for a device indicates a greater disparity in its resource configuration compared to other devices in the network.

\subsubsection{Resource-Aware Performance Prediction Model}

\begin{algorithm}[!t]
\caption{Resource Aware Performance Prediction Model}
\label{alg:performance_prediction_model}
\begin{algorithmic}[1]
\Require Current resource state $R_i(t)$, MoE model config $\mathcal{M}$, learning rate $\eta$
\Ensure Predicted performance metrics ($L_{\text{comp}}, L_{\text{comm}}, M_{\text{usage}}$), updated coefficients $\beta_{ij}$
\State Initialize regression coefficients $\beta_{ij}$
\While{not converged}
    \State Predict metrics: $L_{\text{comp}}(R_i(t), \mathcal{M})$, $L_{\text{comm}}(R_i(t), \mathcal{M})$, $M_{\text{usage}}(R_i(t), \mathcal{M})$
    \State Measure actual performance metrics: $y$
    \State Calculate prediction error loss: $L(y, \hat{y})$
    \State Update coefficients: $\beta_{ij}^{(t+1)} = \beta_{ij}^{(t)} - \eta \cdot \nabla_{\beta_{ij}} L(y, \hat{y})$
    \If{prediction accuracy converges}
        \State \textbf{break}
    \EndIf
\EndWhile
\State \Return Predicted metrics, updated $\beta_{ij}$
\end{algorithmic}
\end{algorithm}

To establish a clear correlation between the dynamic resource states and the performance of MoE models, we develop a performance prediction model. Given the current resource state $R_i(t)$ and the MoE model configuration $\mathcal{M}$ (which includes parameters such as the number of experts, expert distribution, and routing strategy), our model predicts the following key performance metrics:

\begin{itemize}
\item \textbf{Computation Latency} $(L_{\text{comp}}(R_i(t), \mathcal{M}))$: The estimated time required to execute the forward computation of the MoE model.
\item \textbf{Communication Latency} $(L_{\text{comm}}(R_i(t), \mathcal{M}))$: The estimated time taken for transferring expert parameters between different storage layers.
\item \textbf{Memory Usage} $(M_{\text{usage}}(R_i(t), \mathcal{M}))$: The predicted peak memory consumption of the MoE model during its operation.
\end{itemize}

We employ a polynomial regression model to predict these performance indicators. For instance, the computation latency can be modeled as:
\begin{equation}
L_{\text{comp}}(R_i(t), \mathcal{M}) = \sum_{i=1}^{n} \sum_{j=1}^{m} \beta_{ij} \cdot f_i(R_i(t)) \cdot g_j(\mathcal{M}) + \epsilon
\end{equation}

where $f_i(R_i(t))$ represents the feature function derived from the resource state, $g_j(\mathcal{M})$ is the feature function characterizing the model configuration, $\beta_{ij}$ are the regression coefficients, and $\epsilon$ denotes the error term.

To ensure continuous improvement in prediction accuracy, we implement an online learning strategy that updates the prediction model in real-time:
\begin{equation}
\beta_{ij}^{(t+1)} = \beta_{ij}^{(t)} - \eta \cdot \nabla_{\beta_{ij}} L(y, \hat{y})
\end{equation}

where $\eta$ is the learning rate, $L(y, \hat{y})$ is the prediction error loss function, $y$ represents the actual measured performance metric, and $\hat{y}$ is the value predicted by the model.

\subsubsection{Resource-Aware Decision Functions}

\begin{algorithm}[!t]
\caption{Resource Aware Decision Functions}
\label{alg:resource_aware_decision_functions}
\begin{algorithmic}[1]
\Require Resource state $R_i(t)$, memory stability $S_m(t)$, bandwidth stability $S_b(t)$, available GPU memory $m_i^{\text{avail,gpu}}(t)$, total GPU memory $m_i^{\text{gpu}}(t)$, expert frequency $f_{\text{act}}(e_{l,j})$, comm latency $L_{\text{comm}}(R_i(t), e_{l,j})$
\Ensure Aggregation granularity $E'(t)$, offloading priority $P_{\text{offload}}(e_{l,j}, t)$, prefetching threshold $\theta_{\text{pre}}(t)$
\State \textit{// Expert Aggregation Granularity Decision}
\State Calculate $E'(t) = \min\left\{\max\left\{1, \left\lfloor \frac{A_m \cdot m_i^{\text{avail,gpu}}(t)}{S_e \cdot (1 + \beta(1 - S_m(t)))} \right\rfloor \right\}, E\right\}$
\State \textit{// Expert Offloading Priority Function}
\State Calculate $P_{\text{offload}}(e_{l,j}, t) = \gamma \cdot f_{\text{act}}(e_{l,j}) + (1 - \gamma) \cdot \frac{s_{l,j}}{L_{\text{comm}}(R_i(t), e_{l,j})}$
\State \textit{// Expert Prefetching Decision Function}
\State Calculate $\theta_{\text{pre}}(t) = \theta_{\text{base}} \cdot S_b(t) \cdot \left(1 + \delta \cdot \frac{m_i^{\text{avail,gpu}}(t)}{m_i^{\text{gpu}}(t)}\right)$
\State \Return $E'(t), P_{\text{offload}}(e_{l,j}, t), \theta_{\text{pre}}(t)$
\end{algorithmic}
\end{algorithm}

Based on the comprehensive resource state characterization and the developed performance prediction model, we define a set of resource-aware decision functions that provide critical guidance for optimizing expert aggregation and offloading strategies.

\textbf{Expert Aggregation Granularity Decision Function:} This function dynamically determines the optimal granularity for expert aggregation, taking into account the current resource state and its stability:
\begin{equation}
E'(t) = \min\left\{\max\left\{1, \left\lfloor \frac{A_m \cdot m_i^{\text{avail,gpu}}(t)}{S_e \cdot (1 + \beta(1 - S_m(t)))} \right\rfloor \right\}, E\right\} 
\end{equation}

where $E'(t)$ is the target number of experts after aggregation, $m_i^{\text{avail,gpu}}(t)$ is the currently available GPU memory, $S_e$ is the parameter size of a single expert, $S_m(t)$ is the stability indicator of memory resources, $A_m \in (0, 1)$ is the memory allocation ratio, and $\beta$ is the risk aversion coefficient.

\textbf{Expert Offloading Priority Function:} This function assigns an offloading priority to each expert, considering both the current resource state and the estimated communication latency:

\begin{multline}
P_{\text{offload}}(e_{l,j}, t) = \gamma \cdot f_{\text{act}}(e_{l,j}) + (1 - \gamma) \cdot \frac{s_{l,j}}{L_{\text{comm}}(R_i(t), e_{l,j})}
\end{multline}

where $f_{\text{act}}(e_{l,j})$ represents the activation frequency of expert $e_{l,j}$, $L_{\text{comm}}(R_i(t), e_{l,j})$ is the estimated communication latency for loading expert $e_{l,j}$ under the current resource state, and $\gamma$ is a balancing parameter.

\textbf{Expert Prefetching Decision Function:} This function determines whether to prefetch experts by comprehensively evaluating the bandwidth resource state, its stability, and the accuracy of expert activation predictions:
\begin{multline}
\theta_{\text{pre}}(t) = \theta_{\text{base}} \cdot S_b(t) \cdot \left(1 + \delta \cdot \frac{m_i^{\text{avail,gpu}}(t)}{m_i^{\text{gpu}}(t)}\right)
\end{multline}

where $\theta_{\text{pre}}(t)$ is the dynamic prefetching threshold, $\theta_{\text{base}}$ is the base threshold, $S_b(t)$ is the stability indicator of bandwidth resources, and $\delta$ is an adjustment parameter. If the predicted probability $p(e_i^{(l+1)} | x_t)$ of expert $e_i$ being activated exceeds $\theta_{\text{pre}}(t)$, the system initiates prefetching of that expert.

\subsubsection{Theoretical Analysis and Performance Guarantee}

\begin{algorithm}[!t]
\caption{Theoretical Analysis and Performance Guarantee}
\label{alg:theoretical_analysis}
\begin{algorithmic}[1]
\Require Resource state $R_d(t)$, memory resource $m_d^{\text{gpu}}(t)$, prefetching hit rate $P_{\text{hit}}$
\Ensure Proof of reduced inference failure probability and average inference latency
\State \textit{// Proof for reduced inference failure probability}
\State Assume memory resource $m_d^{\text{gpu}}(t)$ follows distribution $\mathcal{D}$
\State Define fixed aggregation granularity strategy: $E'^{\text{fixed}} = \lfloor(A_m \cdot \mu_m) / S_e\rfloor$
\State Calculate $P_{\text{fail}}^{\text{fixed}} = P(m_d^{\text{gpu}}(t) < E'^{\text{fixed}} \cdot S_e + M_{\text{other}})$
\State Define dynamic resource-aware strategy
\State Calculate $P_{\text{fail}}^{\text{dynamic}} = P\left(m_d^{\text{gpu}}(t) < \frac{A_m \cdot m_d^{\text{gpu}}(t)}{1 + \beta(1 - S_m(t))} + M_{\text{other}}\right)$
\State Show $P_{\text{fail}}^{\text{dynamic}} < P_{\text{fail}}^{\text{fixed}}$
\State \textit{// Proof for reduced average inference latency}
\State Define no-prefetching strategy: $L_{\text{no-prefetch}} = T_{\text{comp}} + T_{\text{load}}$
\State Define fixed-threshold prefetching strategy: $L_{\text{fixed}} = T_{\text{comp}} + (1 - P_{\text{hit}}^{\text{fixed}}) \cdot T_{\text{load}}$
\State Define resource-aware dynamic prefetching strategy
\State Show $P_{\text{hit}}^{\text{dynamic}} \geq P_{\text{hit}}^{\text{fixed}}$
\State Conclude $L_{\text{dynamic}} \leq L_{\text{fixed}} < L_{\text{no-prefetch}}$
\State \Return Proofs completed
\end{algorithmic}
\end{algorithm}

This subsection provides theoretical analysis to substantiate the effectiveness of the proposed dynamic resource perception module.

\begin{theorem}
Under a given resource state $R_d(t)$, employing a dynamic resource-aware expert aggregation strategy significantly reduces the inference failure probability compared to a static, fixed aggregation granularity.
\end{theorem}

\begin{proof}
Let us assume that the memory resource $m_d^{\text{gpu}}(t)$ follows a certain distribution $\mathcal{D}$, with a mean of $\mu_m$ and a variance of $\sigma_m^2$. In a fixed aggregation granularity strategy, the number of experts is set to $E'_{\text{fixed}} = \lfloor(A_m \cdot \mu_m) / S_e\rfloor$. Inference failure occurs when the memory usage exceeds the available GPU memory, i.e., $m_d^{\text{used,gpu}}(t) > m_d^{\text{gpu}}(t)$.

For the fixed strategy, the probability of inference failure is given by:
\begin{align}
P_{\text{fail}}^{\text{fixed}} &= P(m_d^{\text{used,gpu}}(t) > m_d^{\text{gpu}}(t)) \nonumber \\
&= P(m_d^{\text{gpu}}(t) < E'_{\text{fixed}} \cdot S_e + M_{\text{other}})
\end{align}

where $M_{\text{other}}$ represents the memory occupied by parameters other than experts.

Conversely, for the dynamic resource-aware strategy, the probability of inference failure is:
\begin{equation}
P_{\text{fail}}^{\text{dynamic}} = P\left(m_d^{\text{gpu}}(t) < \frac{A_m \cdot m_d^{\text{gpu}}(t)}{1 + \beta(1 - S_m(t))} + M_{\text{other}}\right)
\end{equation}

Given that $\beta > 0$ and $S_m(t) < 1$, it follows that:
\begin{equation}
\frac{A_m \cdot m_d^{\text{gpu}}(t)}{1 + \beta(1 - S_m(t))} < A_m \cdot m_d^{\text{gpu}}(t)
\end{equation}

Therefore, $P_{\text{fail}}^{\text{dynamic}} < P_{\text{fail}}^{\text{fixed}}$, which completes the proof.
\end{proof}

\begin{theorem}
In scenarios with fluctuating bandwidth resources, the implementation of a resource-aware expert prefetching strategy leads to a reduction in average inference latency.
\end{theorem}

\begin{proof}
Let $P_{\text{hit}}$ denote the prefetching hit rate, $T_{\text{load}}$ represent the expert loading time, and $T_{\text{comp}}$ signify the computation time. For a strategy that does not employ prefetching, the average latency is simply:
\begin{equation}
L_{\text{no-prefetch}} = T_{\text{comp}} + T_{\text{load}}
\end{equation}

For a fixed-threshold prefetching strategy, assuming a prefetching hit rate of $P_{\text{hit}}^{\text{fixed}}$, the average latency is:
\begin{equation}
L_{\text{fixed}} = T_{\text{comp}} + (1 - P_{\text{hit}}^{\text{fixed}}) \cdot T_{\text{load}}
\end{equation}

In contrast, for a resource-aware dynamic prefetching strategy, the prefetching hit rate is $P_{\text{hit}}^{\text{dynamic}}$. This strategy intelligently adjusts prefetching aggressiveness based on bandwidth resource stability: it increases aggressiveness when bandwidth stability is high and adopts a more conservative approach when stability is low. Consequently, it holds that $P_{\text{hit}}^{\text{dynamic}} \geq P_{\text{hit}}^{\text{fixed}}$.

Thus, $L_{\text{dynamic}} \leq L_{\text{fixed}} < L_{\text{no-prefetch}}$, which concludes the proof.
\end{proof}

\subsection{Expert Aggregation Method with Dynamic Heterogeneous Resource Awareness}

Traditional expert aggregation methods typically involve statically merging expert parameters during pretraining or fine-tuning phases. However, the dynamic nature of resource states in edge environments necessitates more flexible deployment strategies. To enhance the deployment efficiency of MoE models in resource-constrained settings, we propose a resource-aware expert aggregation strategy. This strategy is built upon the "Merge, then Compress" method~\cite{li2023merge} and adopts a "pre-computation + dynamic selection" implementation to adapt to the fluctuating resource conditions in edge environments.

\subsubsection{Expert Similarity Assessment}

To facilitate effective expert merging, accurate assessment of similarity between individual experts is crucial. Given any two experts, $e_i$ and $e_j$, their similarity is quantified across two primary dimensions: parameter space and functional space.

\textbf{Parameter Space Similarity:} This metric evaluates the structural resemblance between experts by computing the cosine similarity of their respective parameter vectors. The formula is given by:
\begin{equation}
S_{\text{param}}(e_i, e_j) = \frac{\vec{e_i} \cdot \vec{e_j}}{\|\vec{e_i}\| \cdot \|\vec{e_j}\|}
\end{equation}

where $\vec{e_i}$ and $\vec{e_j}$ represent the parameter vectors of expert $e_i$ and $e_j$, respectively.

\textbf{Functional Space Similarity:} This metric assesses the behavioral congruence of experts by comparing their output distributions on a carefully selected calibration dataset. The functional similarity is quantified using the Kullback-Leibler (KL) divergence, as follows:
\begin{equation}
S_{\text{func}}(e_i, e_j) = 1 - \frac{1}{|\mathcal{D}|} \sum_{x \in \mathcal{D}} d_{\text{KL}}(e_i(x) \| e_j(x))
\end{equation}

where $d_{\text{KL}}$ denotes the KL divergence, and $\mathcal{D}$ is the calibration dataset. A smaller KL divergence indicates higher functional similarity.

To provide a comprehensive measure of similarity, both parameter and functional space similarities are combined into a single, integrated metric:
\begin{equation}
S(e_i, e_j) = \alpha \cdot S_{\text{param}}(e_i, e_j) + (1 - \alpha) \cdot S_{\text{func}}(e_i, e_j)
\end{equation}

Here, $\alpha$ is a balancing parameter ($0 \leq \alpha \leq 1$) that can be adjusted based on the specific requirements and characteristics of the application scenario.

\subsubsection{Frequency-Based Expert Aggregation Strategy}

Our framework incorporates two variants of the "Merge, then Compress" expert fusion approach~\cite{li2023merge}: fixed-ratio fusion and adaptive-ratio fusion. Both methods leverage expert activation frequency and similarity analysis but differ in their strategy for determining the aggregation ratio for each MoE layer. The aggregation process involves three key steps:

\textbf{Step 1: Principal Expert Identification.} This step involves identifying the most frequently activated experts within each MoE layer. Using a carefully curated calibration dataset, the activation frequency of each expert $e_{l,j}$ in layer $l$ is calculated as:
\begin{equation}
f_{\text{act}}(e_{l,j}) = \frac{\text{Number of times expert } e_{l,j} \text{ is activated}}{\text{Total activations in layer } l}
\end{equation}

Experts whose activation frequency exceeds a predefined threshold, $\theta_{\text{act}}$, are designated as principal experts for that layer. These experts are considered crucial for maintaining the model's core functionality.

\textbf{Step 2: Expert Grouping.} Following the identification of principal experts, secondary experts are grouped with the most functionally and parametrically similar principal expert. If the set of principal experts for layer $l$ is $\{e_{l,m_1}, e_{l,m_2}, \ldots, e_{l,m_k}\}$, then each secondary expert $e_{l,j}$ is assigned to group $g_i$, where $i$ is determined by:
\begin{equation}
i = \arg\max_{1 \leq r \leq k} S(e_{l,j}, e_{l,m_r})
\end{equation}

This ensures that experts with similar functionalities are grouped together to facilitate effective merging.

\textbf{Step 3: Frequency-Based Merging.} For each identified group $g_i = \{e_{l,m_i}, e_{l,j_1}, e_{l,j_2}, \ldots\}$, a weighted averaging operation is performed to merge the experts. The merged expert $\hat{e}_{l,i}$ is computed as:
\begin{equation}
\hat{e}_{l,i} = \frac{\sum_{e_{l,j} \in g_i} f_{\text{act}}(e_{l,j}) \cdot e_{l,j}}{\sum_{e_{l,j} \in g_i} f_{\text{act}}(e_{l,j})}
\end{equation}

This weighted averaging strategy ensures that experts with higher activation frequencies exert greater influence on the merged expert, thereby preserving the model's critical functionalities and performance.

\textbf{Fixed-Ratio Fusion:} In this variant, a predetermined, fixed proportion of experts is retained in each MoE layer after aggregation. For a layer originally containing $E$ experts, fixed-ratio fusion retains $E' = \lfloor r \cdot E \rfloor$ experts, where $r$ is the predefined retention ratio (e.g., $r = 0.5$ for retaining half the experts).

\textbf{Adaptive-Ratio Fusion:} This method dynamically determines the number of experts to retain for each layer, adapting to the unique characteristics of the expert activation distribution within that layer. We quantify the expert activation distribution using the expert activation entropy for layer $l$:
\begin{equation}
H_l = -\sum_{j=1}^{E} f_{\text{act}}(e_{l,j}) \cdot \log f_{\text{act}}(e_{l,j})
\end{equation}

To normalize the entropy across different layers, we use the normalized entropy:
\begin{equation}
\bar{H}_l = \frac{H_l}{\log E}
\end{equation}

The number of retained experts for layer $l$, $E'_l$, is then calculated based on this normalized entropy:
\begin{equation}
E'_l = \max\left\{E_{\text{min}}, \left\lfloor E \cdot (r_{\text{base}} + \Delta r \cdot \bar{H}_l) \right\rfloor\right\}
\end{equation}

where $E_{\text{min}}$ is the minimum number of experts to retain, $r_{\text{base}}$ is a base retention ratio, and $\Delta r$ is an adjustment parameter.

\subsubsection{Model Architecture Specialization}

For MoE models employing Encoder-Decoder architectures, such as Switch-base-32, we implement an architecture-specific optimization strategy. This approach acknowledges the distinct computational characteristics and resource consumption patterns of the Encoder and Decoder components, necessitating differentiated optimization strategies.

\textbf{Encoder-Focused Optimization:} Our primary focus for applying expert aggregation and offloading strategies is the Encoder part of the Switch-base-32 model. This prioritization is based on the following key considerations:

\begin{itemize}
\item The Encoder is responsible for processing input sequences, which typically involves a higher computational load and is inherently more amenable to parallel processing techniques.
\item The Decoder operates in an autoregressive manner, generating tokens sequentially. This token-by-token execution leads to a high frequency of expert scheduling, where the overhead associated with offloading might potentially negate any performance benefits.
\item Empirical studies have consistently shown that optimizing the Encoder part yields the most significant efficiency improvements while effectively preserving the overall model performance.
\end{itemize}

\subsubsection{Resource-Aware Pre-computed Model Selection}

Leveraging the expert fusion strategies outlined above, we pre-compute and store multiple model versions, each configured with different aggregation settings, prior to deployment. For the Switch-base-32 model, our pre-computed model library encompasses the following configurations:

\begin{itemize}
\item \textbf{Fixed-Ratio Fusion Series:}
\begin{itemize}
\item \textbf{Light Fusion Version:} Retains 24 experts per layer ($r = 0.75$).
\item \textbf{Medium Fusion Version:} Retains 16 experts per layer ($r = 0.5$).
\item \textbf{Heavy Fusion Version:} Retains 8 experts per layer ($r = 0.25$).
\end{itemize}

\item \textbf{Adaptive-Ratio Fusion Series:}
\begin{itemize}
\item \textbf{Light Adaptive Version:} $r_{\text{base}} = 0.6$, $\Delta r = 0.3$.
\item \textbf{Medium Adaptive Version:} $r_{\text{base}} = 0.4$, $\Delta r = 0.2$.
\item \textbf{Heavy Adaptive Version:} $r_{\text{base}} = 0.2$, $\Delta r = 0.1$.
\end{itemize}
\end{itemize}

During runtime, the resource perception module continuously monitors the real-time resource status of the edge devices. Based on the current available resources, $M_{\text{avail}}(t)$, and the prevailing performance requirements, the system dynamically selects the optimal model version from the pre-computed library. The selection criterion is formulated as:
\begin{equation}
\mathcal{M}^* = \arg\max_{\mathcal{M}_j \in \mathcal{M}_{\text{lib}}} P(\mathcal{M}_j) \quad \text{s.t.} \quad M_{\text{req}}(\mathcal{M}_j) \leq M_{\text{avail}}(t)
\end{equation}

where $P(\mathcal{M}_j)$ represents the performance metric of model $\mathcal{M}_j$, $M_{\text{req}}(\mathcal{M}_j)$ denotes its resource requirements, and $\mathcal{M}_{\text{lib}}$ is the set of all precomputed model versions.

To mitigate the overhead associated with frequent model version switching, we incorporate a mechanism that evaluates the stability of resource states and the potential costs of switching. A model version switch is executed only when two conditions are met: (1) the resource state exhibits a significant and sustained change, and (2) the anticipated performance benefits of switching demonstrably outweigh the associated costs. This condition is expressed as:
\begin{equation}
\Delta P(\mathcal{M}_{\text{current}}, \mathcal{M}^*) > \lambda \cdot C_{\text{switch}} \quad \text{and} \quad T_{\text{stable}} > T_{\text{threshold}}
\end{equation}

where $\Delta P$ represents the performance improvement from switching, $\lambda$ is a weighting factor, $C_{\text{switch}}$ is the cost of switching, $T_{\text{stable}}$ is the duration for which the new resource state has been stable, and $T_{\text{threshold}}$ is a minimum stability duration.

\subsection{Expert Activation Prediction and Multi-Level Storage Strategy}

Building upon the selected expert aggregation model, we further develop an efficient expert offloading strategy to facilitate seamless and efficient collaboration between GPU and CPU memory. Consistent with our architecture-specific optimization approach, for Encoder-Decoder models like Switch-base-32, expert prediction and offloading strategies are primarily applied to the Encoder component.

\subsubsection{Expert Activation Prediction}

To minimize the latency associated with expert loading, we design a lightweight prediction model specifically for pre-determining expert activation patterns. This model is trained using a combination of features, including:

\begin{itemize}
\item \textbf{Current Layer Activation Features:} The set of experts selected in the current layer, denoted as $\{e_1^l, e_2^l, \ldots, e_K^l\}$.
\item \textbf{Input Features:} The embedding representation of the current input token, $h_t$.
\item \textbf{Context Features:} Additional contextual information, such as positional embeddings and historical expert activation patterns, denoted as $c_t$.
\end{itemize}

The prediction model employs a two-layer multilayer perceptron (MLP) structure, which outputs the probability of expert $e_i$ being activated in the $(l+1)$-th layer for a given token $x_t$:

\begin{multline}
p(e_i^{(l+1)} | x_t) = \text{Softmax}(W_2 \cdot \text{ReLU}(W_1 \cdot [e_1^l, \ldots, e_K^l, h_t, c_t] \\
+ b_1) + b_2)  
\end{multline}

where $W_1$ and $W_2$ are weight matrices, $b_1$ and $b_2$ are bias vectors, and ReLU is the rectified linear unit activation function.

\subsubsection{Multi-Level Storage Collaborative Scheduling}

Based on the expert activation prediction results, we propose a multilevel storage collaborative scheduling strategy, as conceptually illustrated in Figure~\ref{fig:multi_level_storage}. This strategy partitions device memory into three distinct hierarchical levels to optimize expert management:

\begin{figure}[!t]
\centering
\includegraphics[width=\columnwidth]{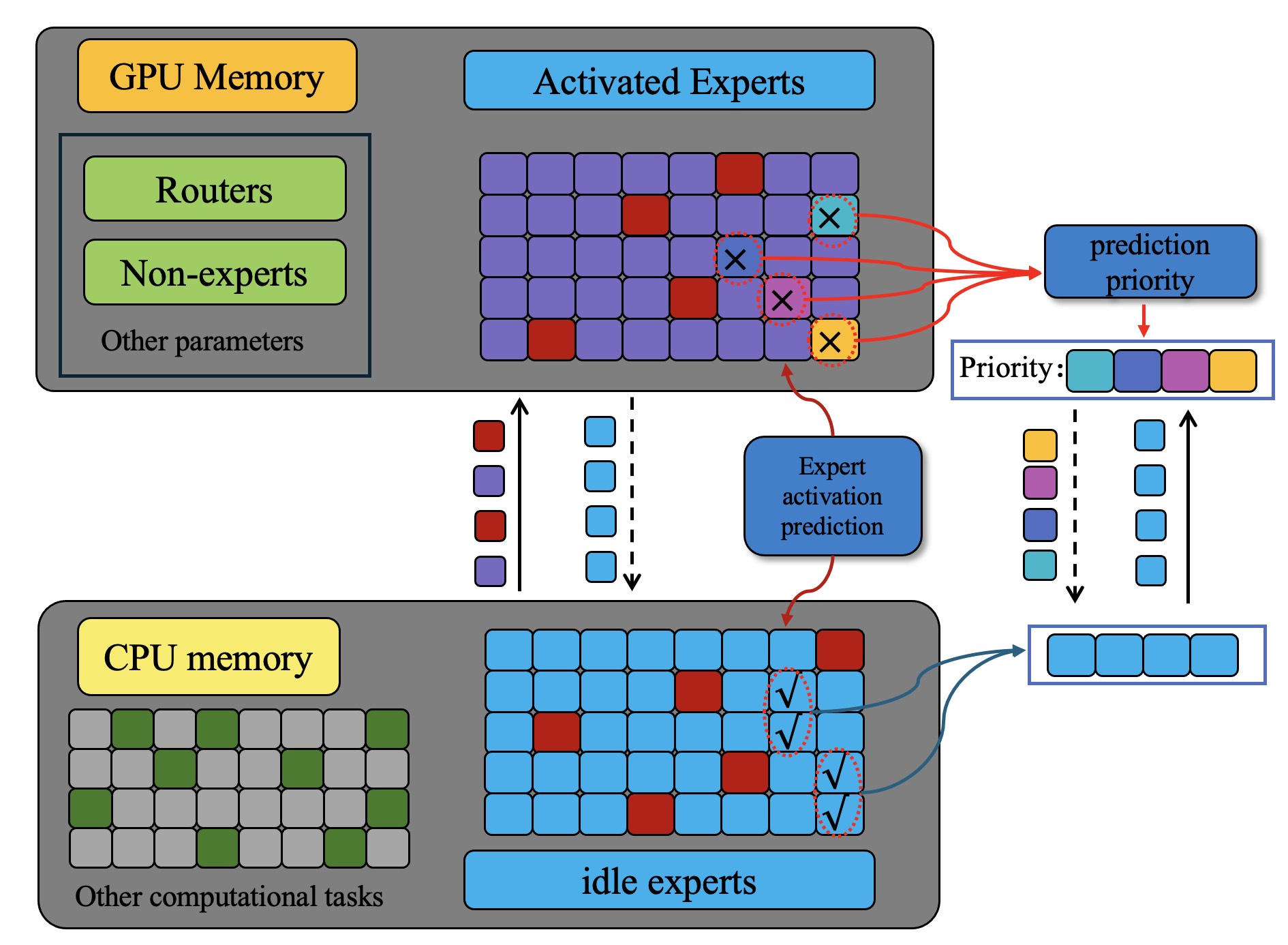}
\caption{Multilevel Storage Collaborative Scheduling Strategy}
\label{fig:multi_level_storage}
\end{figure}

\begin{itemize}
\item \textbf{GPU Workspace:} This area is designated for storing currently active experts and non-expert parameters that are immediately required for computation.
\item \textbf{GPU Cache Area:} This region serves as a temporary buffer for experts that have a high predicted probability of being activated in the near future, enabling rapid access.
\item \textbf{CPU Memory:} This constitutes the primary storage for all remaining experts that are not currently in the GPU workspace or cache.
\end{itemize}

The dynamic scheduling of experts between these different storage layers adheres to the following rules:

\begin{itemize}
\item \textbf{Prefetching Strategy:} The prefetching threshold, $\theta_{\text{pre}}$, is adaptively adjusted based on the output probability $p(e_i^{(l+1)} | x_t)$ from the prediction model. This adaptive threshold is defined as:
\begin{equation}
\theta_{\text{pre}} = \gamma \cdot \frac{M_{\text{GPU}}^{\text{avail}}}{M_{\text{GPU}}}
\end{equation}
where $\gamma$ is a base prefetching threshold, $M_{\text{GPU}}^{\text{avail}}$ represents the currently available GPU memory, and $M_{\text{GPU}}$ denotes the total GPU memory capacity. If the predicted activation probability $p(e_i^{(l+1)} | x_t)$ for expert $e_i$ exceeds $\theta_{\text{pre}}$, expert $e_i$ is proactively prefetched into the GPU cache area.

\item \textbf{Eviction Strategy:} We introduce a hybrid cache replacement policy that comprehensively considers multiple factors: access frequency, expert importance, and predicted activation probability. The eviction score $S_{\text{evict}}(e_i)$ for an expert $e_i$ is calculated as:
\begin{align}
S_{\text{evict}}(e_i) &= (1 - \delta) \cdot (1 - p(e_i^{(l+1)} | x_t)) \nonumber \\
&\quad + \delta \cdot \left(\lambda \cdot \frac{1}{f_{\text{recent}}(e_i)} + (1 - \lambda) \cdot \frac{1}{I(e_i)}\right)
\end{align}
where $f_{\text{recent}}(e_i)$ denotes the recent access frequency of expert $e_i$, $I(e_i)$ represents the importance of expert $e_i$ (a predefined or learned metric), and $\delta$ and $\lambda$ are weighting parameters ($0 \leq \delta, \lambda \leq 1$). Experts with the highest eviction score are prioritized for offloading from the GPU to make space for more critical or frequently accessed experts.
\end{itemize}

\subsection{Expert Offloading and Scheduling Strategy Based on Heterogeneous Resource Collaboration}

Building upon expert aggregation techniques, we further refine resource management by designing an expert offloading and scheduling strategy that leverages heterogeneous resource collaboration within the edge server. This strategy aims to minimize GPU memory demands and communication overhead, thereby enhancing the utilization efficiency of the edge server's heterogeneous resources through adaptive dynamic caching and scheduling between GPU and CPU memory.

\subsubsection{Adaptive Dynamic Caching Based on Multi-Dimensional Expert Information Fusion and Intra-Server Heterogeneous Resource Collaboration}

Our approach utilizes historical token computation data from the edge server $S$ to estimate the popularity, $L_H^S = \{l_H^{(1:1)}, l_H^{(1:2)}, \ldots, l_H^{(m:n)}\}$, of different experts, $E_H = \{e_H^{(1:1)}, e_H^{(1:2)}, \ldots, e_H^{(m:n)}\}$, within the Mixture-of-Experts model $H$. Here, $n$ represents the number of experts per layer, and $m$ denotes the number of expert layers, resulting in a total of $N = n \times m$ experts in model $H$. Based on the available GPU memory, $G_S$, of the edge server $S$, the system dynamically determines the number of experts, $\eta_{G_S}$, which can be loaded into GPU memory. The remaining $N - \eta_{G_S}$ experts are stored in the CPU memory.

The caching status of experts between GPU and CPU memory is dynamically adjusted based on two primary factors: (1) the real-time fluctuations in the available GPU and CPU memory resources of the edge server, and (2) the varying requirements for specific experts across different MoE layers during token forward computation.

To further enhance caching efficiency, we employ an expert activation prediction algorithm. This algorithm forecasts the expert activations for the first $Z$ layers required for the initial $N$ tokens in the input sequence, yielding $M$ predicted experts that are likely to be activated. For these predicted activated experts, we prioritize dynamic cache replacement according to the following hierarchy:

\begin{itemize}
\item Experts with high activation frequency that are not currently cached in GPU memory are given top priority for transfer from CPU memory to GPU memory.
\item Experts that are anticipated to be used in the immediate next inference step are prioritized for caching in GPU memory.
\item Experts with low usage frequency in GPU memory and those not included in the predicted activation set are prioritized for replacement and offloading.
\end{itemize}

Recognizing that expert activation predictions are not infallible, we integrate expert popularity, expert importance, and expert similarity to establish a robust error correction prioritization mechanism. For predicted erroneous experts with high priority, immediate cache replacement is performed. For predicted erroneous experts with lower priority, we attempt to substitute them with similar experts already cached in GPU memory, leveraging expert similarity to reduce the frequency of costly data transfers.

\subsection{Collaborative Optimization Framework}

Building upon the previously detailed expert aggregation and offloading methodologies, we propose a comprehensive collaborative optimization framework. This framework synergistically integrates both techniques, establishing a two-level resource adaptation mechanism. A key aspect of this framework is its specialized optimization strategy for the Encoder part within Encoder-Decoder architectures.

\subsubsection{Two-Level Resource Adaptation Mechanism}

Our framework employs a hierarchical approach to resource adaptation, addressing resource fluctuations at different temporal granularities:

\begin{itemize}
\item \textbf{Coarse-grained Resource Adaptation:} This level addresses medium-to-long-term resource changes (on the order of seconds to minutes) by dynamically selecting pre-computed model versions with varying aggregation configurations. This includes:
\begin{itemize}
\item \textbf{Fixed-Ratio Fusion Versions:} These versions offer deterministic resource requirements and predictable performance characteristics, making them suitable for stable resource conditions.
\item \textbf{Adaptive-Ratio Fusion Versions:} These versions provide superior performance but exhibit more variable resource requirements, making them ideal for dynamic environments where performance is paramount.
\end{itemize}

\item \textbf{Fine-grained Resource Adaptation:} This level handles short-term resource fluctuations (on the order of milliseconds) through dynamic expert offloading and prefetching strategies. This mechanism is primarily applied to the Encoder part of the MoE model and involves:
\begin{itemize}
\item Prioritizing the loading of critical experts based on real-time expert prediction results.
\item Optimizing expert caching and scheduling to minimize latency and maximize resource utilization.
\end{itemize}
\end{itemize}

\subsubsection{Collaborative Decision-Making Process}

\begin{algorithm}[!t]
\caption{Collaborative Decision Making Process}
\label{alg:collaborative_decision_making}
\begin{algorithmic}[1]
\Require Global objective $O$, local objectives $O_i$, resource states $R_i(t)$, predictions $P_i(t)$
\Ensure Optimal expert offloading and aggregation strategy $S^*$
\State Initialize global model $G$
\State Initialize local models $L_i$ for each device $d_i$
\While{not converged}
    \State \textit{// Local Optimization Phase}
    \For{each device $d_i$ in parallel}
        \State $d_i$ collects $R_i(t)$ and $P_i(t)$
        \State $d_i$ optimizes local objective $O_i$ to propose local strategy $s_i$
        \State $d_i$ sends $s_i$ to global orchestrator
    \EndFor
    \State \textit{// Global Aggregation and Refinement Phase}
    \State Global orchestrator aggregates all $s_i$
    \State Global orchestrator refines aggregated strategy based on $O$
    \State Global orchestrator broadcasts refined strategy $S_{\text{refined}}$ to all $d_i$
    \State \textit{// Iterative Adjustment}
    \For{each device $d_i$ in parallel}
        \State $d_i$ adjusts its local strategy based on $S_{\text{refined}}$
    \EndFor
    \If{convergence criteria met}
        \State \textbf{break}
    \EndIf
\EndWhile
\State \Return $S^*$
\end{algorithmic}
\end{algorithm}

The collaborative optimization framework operates through a well-defined decision-making process, encompassing initialization, runtime, and feedback adjustment phases:

\begin{itemize}
\item \textbf{Initialization Phase:}
\begin{itemize}
\item The resource perception module conducts an initial assessment of the edge device's capabilities.
\item Based on this initial resource state, an appropriate pre-computed aggregated model version is selected from the model library.
\item The selected model is loaded onto the device, and the initial expert distribution (between GPU and CPU memory) is established.
\end{itemize}

\item \textbf{Runtime Phase:}
\begin{itemize}
\item The resource perception module continuously monitors the realtime resource status of the edge device.
\item If significant and persistent changes in resource conditions are detected:
\begin{itemize}
\item The system re-evaluates the optimal model version, considering the benefits and costs associated with switching.
\item If the performance improvement from switching outweighs the cost ($\Delta P(\mathcal{M}_{\text{current}}, \mathcal{M}^*) > \lambda \cdot C_{\text{switch}}$) and the new resource state has been stable for a sufficient duration ($T_{\text{stable}} > T_{\text{threshold}}$), a model version switch is executed.
\end{itemize}
\item Based on the currently loaded aggregated model:
\begin{itemize}
\item Expert prediction and prefetching are performed specifically for the Encoder part of the model.
\item Expert caching and scheduling are continuously optimized.
\item Computation-communication collaborative optimization is performed to ensure efficient resource utilization.
\end{itemize}
\end{itemize}

\item \textbf{Feedback Adjustment:}
\begin{itemize}
\item Actual operational performance data is collected during runtime.
\item This data is used to update and refine the resource prediction model.
\item Decision thresholds and parameters within the framework are adjusted to further enhance adaptive capabilities.
\end{itemize}
\end{itemize}

Algorithm~\ref{alg:collaborative_optimization} provides a high-level overview of the collaborative optimization framework's main process.

\begin{algorithm}[!t]
\caption{Expert Aggregation and Offloading Collaborative Optimization Algorithm}
\label{alg:collaborative_optimization}
\begin{algorithmic}[1]
\Require Precomputed model library $\mathcal{M}_{\text{lib}}$, input sequence $X$, resource constraints $R_i(t)$
\Ensure Model output $Y$
\State Initial resource assessment $R_i(0)$
\State Select initial model version: 
\Statex \hspace*{\algorithmicindent} $\mathcal{M}_{\text{init}} = \arg\max_{\mathcal{M}_j \in \mathcal{M}_{\text{lib}}} P(\mathcal{M}_j)$ 
\Statex \hspace*{\algorithmicindent} subject to $M_{\text{req}}(\mathcal{M}_j) \leq M_{\text{avail}}(0)$
\State Load model $\mathcal{M}_{\text{init}}$ to device, init expert distribution
\For{\textbf{each} token $x_t$ \textbf{in} input sequence $X$}
    \State Update resource state $R_i(t)$
    \If{resource conditions change significantly}
        \State Reevaluate optimal model version $\mathcal{M}^*$
        \If{$\Delta P(\mathcal{M}_{\text{current}}, \mathcal{M}^*) > \lambda \cdot C_{\text{switch}}$}
            \State Switch to model version $\mathcal{M}^*$
        \EndIf
    \EndIf
    \For{\textbf{each} layer $l = 1$ \textbf{to} $L_E$ \textbf{in} Encoder}
        \State Predict next layer activated experts
        \State Perform expert prefetching
        \State Optimize expert scheduling
        \State Perform forward computation
    \EndFor
    \For{\textbf{each} layer $l = 1$ \textbf{to} $L_D$ \textbf{in} Decoder}
        \State Perform standard forward computation
    \EndFor
\EndFor
\State \Return model output $Y$
\end{algorithmic}
\end{algorithm}

Through this collaborative optimization framework, we achieve efficient deployment of MoE models in resource-constrained and dynamically changing edge environments. This framework effectively balances computational performance with resource adaptability. The specialized processing for Encoder-Decoder architectures, coupled with the flexible selection between fixed-ratio and adaptive-ratio fusion strategies, enables the system to provide optimal deployment solutions tailored to actual resource conditions and performance requirements.

\section{EXPERIMENTAL EVALUATION}

To comprehensively evaluate the efficacy of the proposed collaborative optimization framework, we meticulously designed and conducted a series of experiments. This section details the experimental setup, evaluation metrics, and presents a thorough analysis of the obtained results. Our primary objectives were to address the following key research questions:

\begin{enumerate}
    \item How does the proposed collaborative optimization framework perform in terms of latency and memory utilization compared to existing expert aggregation or expert offloading methods?
    \item What is the adaptability and robustness of the collaborative optimization framework under various resource-constrained scenarios?
    \item How do fixed-ratio fusion and adaptive-ratio fusion perform with varying numbers of retained experts?
    \item What is the effect of applying the expert offloading strategy solely to the Encoder part of the model?
    \item How do the synergistic effects of expert aggregation and offloading contribute to adaptability and efficiency across different resource scenarios?
    \item What is the practical applicability of the collaborative optimization framework on real world edge devices?
\end{enumerate}

\subsection{Experimental Setup}

\subsubsection{Models and Datasets}

Our evaluation was conducted on the following Mixture-of-Experts (MoE) models:

\begin{itemize}
    \item \textbf{Switch Transformer Series:}
    \begin{itemize}
        \item \textbf{Switch-Base-8:} Comprising 12 Transformer layers, with odd-numbered layers (1, 3, 5, 7, 9, 11) as MoE layers. Each MoE layer contains 8 experts, totaling approximately 0.5 billion parameters (0.5B), using a Top-1 routing strategy.
        \item \textbf{Switch-Base-32:} Similar to Switch-Base-8, but with 32 experts per MoE layer, resulting in approximately 1.98 billion parameters (1.98B), also using Top-1 routing.
        \item \textbf{Switch-Base-64:} Featuring 64 experts per MoE layer, with a total of approximately 3.8 billion parameters (3.8B), using Top-1 routing.
        \item \textbf{Switch-Base-128:} Equipped with 128 experts per MoE layer, accumulating approximately 7.4 billion parameters (7.4B), using Top-1 routing.
        \item \textbf{Switch-Base-256:} The largest in this series, with 256 experts per MoE layer, totaling approximately 15.8 billion parameters (15.8B), using Top-1 routing.
    \end{itemize}
\end{itemize}

The Switch Transformer series, proposed by Google Research \cite{google-switch}, features a unique alternating structure where only odd layers are MoE layers. Each MoE layer employs a Top-1 routing strategy, activating only one expert per token, while even layers are standard dense Transformer layers. This design significantly reduces computational complexity while maintaining model capacity. Notably, the Switch-Base-32 model (with approximately 1.98B total parameters, but activating only one expert per inference) strikes an excellent balance between computational efficiency and model capability. It significantly outperforms T5-Base (220M parameters) and can approach or exceed the performance of T5-Large (770M parameters) on various tasks, while requiring substantially fewer computational resources than these dense models. The Switch-Base series models maintain a consistent underlying architecture, differing only in the number of experts within their MoE layers, which allows for systematic analysis of the impact of expert scale on optimization effectiveness.

To ensure comprehensive evaluation, we selected the following datasets for testing:

\begin{itemize}
    \item \textbf{SST2:} A sentiment classification dataset for evaluating model performance on binary sentiment analysis tasks.
    \item \textbf{MRPC:} A paraphrase identification dataset for assessing a model's ability to determine semantic equivalence between sentence pairs.
    \item \textbf{COLA:} A linguistic acceptability judgment dataset for evaluating grammatical acceptability, part of the GLUE benchmark.
    \item \textbf{MultiRC:} A multiple-choice question-answering dataset that tests reading comprehension and reasoning capabilities.
    \item \textbf{WikiQA:} A closed-book question-answering dataset that tests knowledge retrieval and inference abilities.
    \item \textbf{HotpotQA:} A multi-hop reasoning question-answering dataset, testing the model's complex reasoning abilities.
    \item \textbf{COPA:} A sentence completion dataset for evaluating causal reasoning and common sense understanding.
    \item \textbf{WinoGrande:} A conference resolution dataset that challenges models with pronoun disambiguation tasks requiring commonsense reasoning.
\end{itemize}
In our experiments, we randomly selected 500 samples from each dataset to calculate average performance metrics. For inference latency evaluation, we set the batch size to 1 and sequence lengths to 128, 512, and 1024, to simulate realistic application scenarios.

\subsubsection{Hardware and Software Configuration}

Experiments were conducted in the following hardware environments:

\begin{itemize}
    \item \textbf{Edge Devices:}
    \begin{itemize}
        \item NVIDIA Jetson AGX Orin (64GB RAM, 12-core ARM CPU, 32GB GPU memory)
        \item Intel NUC 13 Pro (32GB RAM, Intel Core i7-1360P, NVIDIA RTX A500 8GB GPU)
    \end{itemize}
    \item \textbf{Server Device:}
    \begin{itemize}
        \item NVIDIA A100 (80GB) GPU, AMD EPYC 7742 CPU, 256GB RAM
    \end{itemize}
\end{itemize}

The software environment included PyTorch 2.1.0, CUDA 12.1, DeepSpeed 0.10.0, and TensorRT 8.6. 

\subsubsection{Evaluation Metrics}

We employed the following metrics to assess performance:

\begin{itemize}
    \item \textbf{Latency Metrics:} Average inference latency (ms), throughput (tokens/s), and end to end generation time (s).
    \item \textbf{Memory Metrics:} Peak GPU memory usage (MB) and average GPU memory utilization (
    \item \textbf{Performance Metrics:} Accuracy (for SST2, MRPC, MultiRC, COPA, WinoGrande), Exact Match and F1 score (for WikiQA, HotpotQA).
    \item \textbf{Communication Metrics:} Number of expert loads and total GPU-CPU communication volume (GB).
\end{itemize}

\subsection{Expert Aggregation Experiment}

\subsubsection{Expert Fusion Experiment}

We first evaluated the performance of both fixed ratio fusion and adaptive ratio fusion on the Switch-Base-32 model. 

Here are several key observations:

\begin{itemize}
    \item \textbf{Expert Distribution:} Detailed analysis of model parameters revealed that the adaptive fusion strategy allocated varying numbers of experts across different layers. Specifically, adaptive fusion retained between 11-16 experts in the encoder layers, while being more aggressive in the decoder layers, retaining only 3-5 experts.
    \item \textbf{Parameter Efficiency:} The adaptive fusion strategy achieved performance comparable to, or superior to, the original model with only 0.73B parameters (36.9\% of the original model).
    \item \textbf{Fixed vs. Adaptive Comparison:} Compared to fixed-ratio r=0.25 (0.62B) which has a similar parameter count, adaptive fusion demonstrated significant advantages across all tasks.
    \item \textbf{Layer-wise Differences:} While the fixed ratio strategy uniformly applied the same fusion ratio across all layers, the adaptive strategy could capture the varying importance of different layers.
\end{itemize}

\FloatBarrier

Figure~\ref{fig:param_perf_tradeoff} illustrates the parameter-performance trade-off relationship under different expert fusion strategies.

\begin{figure}[!t]
\centering
\includegraphics[width=0.8\columnwidth]{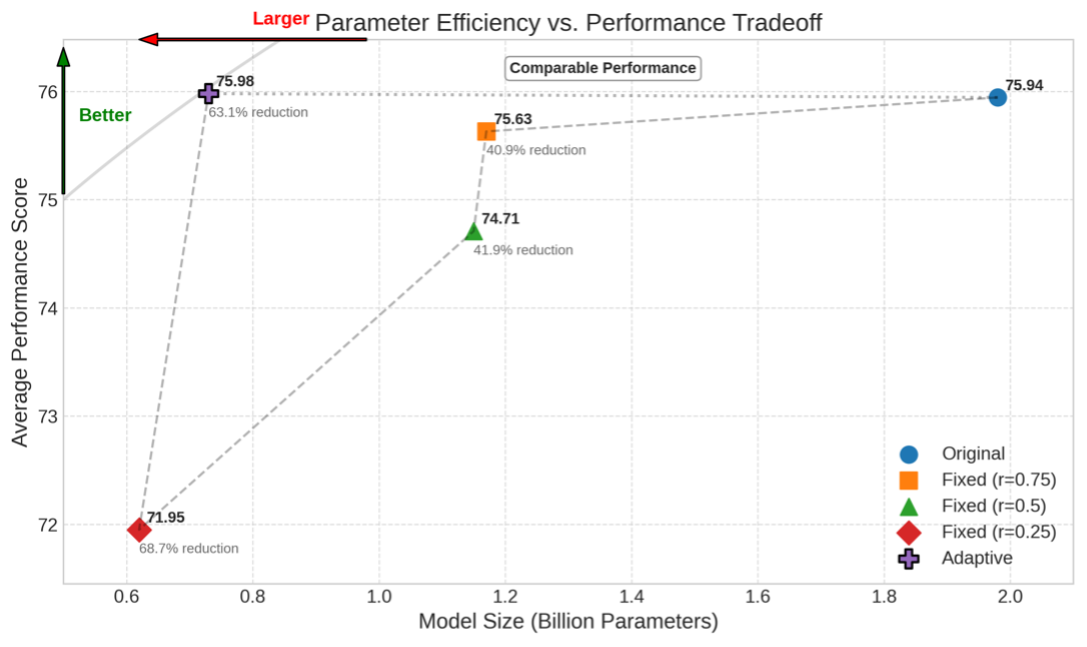}
\caption{Parameter-Performance Trade-off for Different Expert Fusion Strategies}
\label{fig:param_perf_tradeoff}
\end{figure}

To understand the behavior of the adaptive fusion strategy across different layers, we analyzed the expert retention in the six MoE layers of the Switch-Base-32 model, as shown in Figure~\ref{fig:expert_retention}.

\FloatBarrier

\begin{figure}[!t]
\centering
\includegraphics[width=0.8\columnwidth]{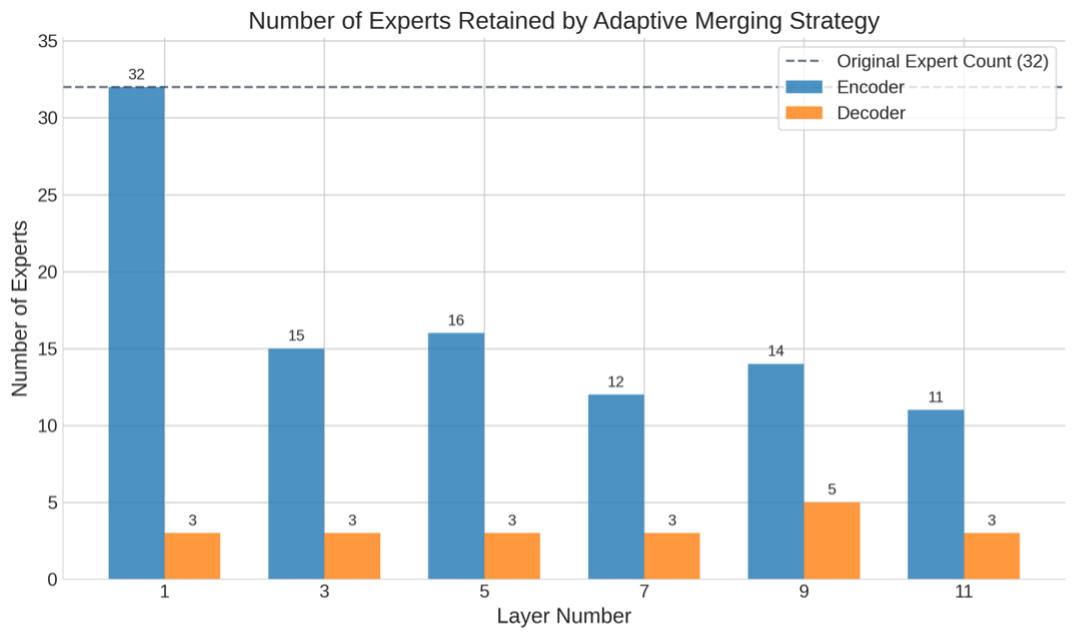}
\caption{Number of Experts Retained by Adaptive Fusion Strategy in Different Layers of Switch-Base-32}
\label{fig:expert_retention}
\end{figure}

As depicted in Figure \ref{fig:expert_retention}, the adaptive fusion strategy allocates varying numbers of experts to different layers. Specifically, layers 3 and 15 retained more experts, which aligns with our analysis: the expert activation distribution in these two layers is more uniform, indicating that more experts contribute significantly to performance. In contrast, layers 7 and 11 retained fewer experts due to a pronounced "hot expert" phenomenon, where a minority of experts accounted for the majority of activations.

It is particularly noteworthy that although the average number of experts for the adaptive fusion strategy (13.5) was slightly lower than that of fixed ratio r=0.5 (16), its performance was significantly better. This indicates that intelligent allocation of expert resources is more effective than simple uniform reduction. By retaining more experts in critical layers and performing more aggressive merging in redundant layers, adaptive fusion can significantly reduce the number of parameters while maintaining or improving model performance.

\subsection{Performance Evaluation of Switch Series Models}

To comprehensively assess the effectiveness of the proposed collaborative optimization framework across different model scales, we conducted detailed performance tests on the Switch Transformer series models. For each model scale, we compared four configurations: the original model (Original), expert aggregation only (M-SMoE), expert offloading only (E-Offload), and collaborative optimization combining both (CoMoE).

\subsubsection{Switch-Base-8 Model Evaluation}

Switch-Base-8, the smallest model in the Switch series, has approximately 0.5B parameters. 
Figure \ref{fig:sb8_latency_throughput} illustrates the latency and throughput comparison for the Switch-Base-8 model across various configurations.

\begin{figure}[!t]
\centering
\includegraphics[width=0.8\columnwidth]{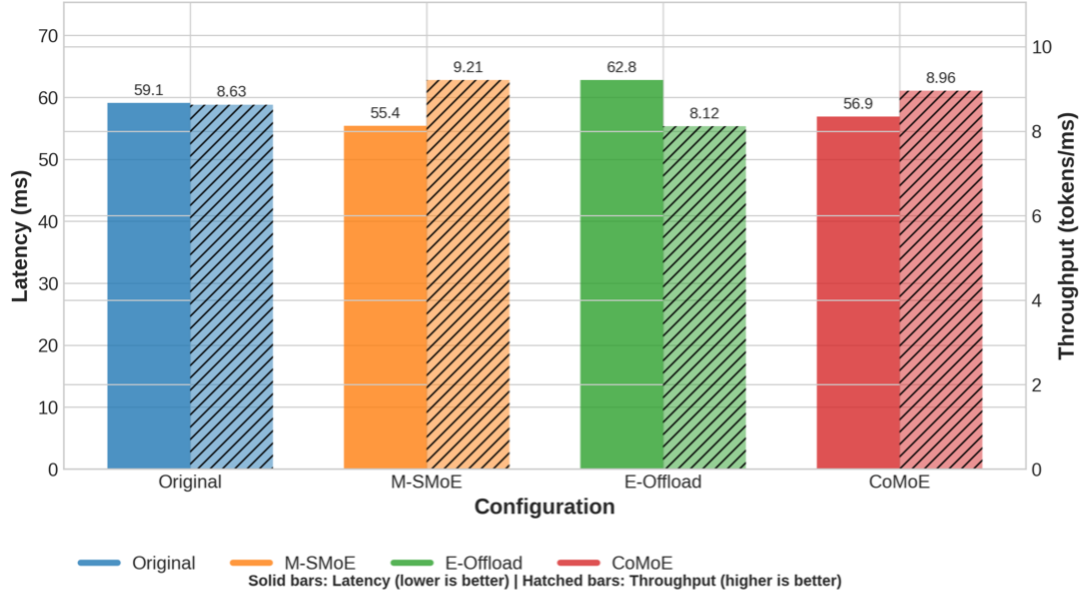}
\caption{Latency and Throughput Comparison of Switch-Base-8 Model}
\label{fig:sb8_latency_throughput}
\end{figure}

Figure \ref{fig:sb8_memory_params} displays the memory footprint and parameter count comparison for the Switch-Base-8 model under different configurations.

\begin{figure}[!t]
\centering
\includegraphics[width=0.8\columnwidth]{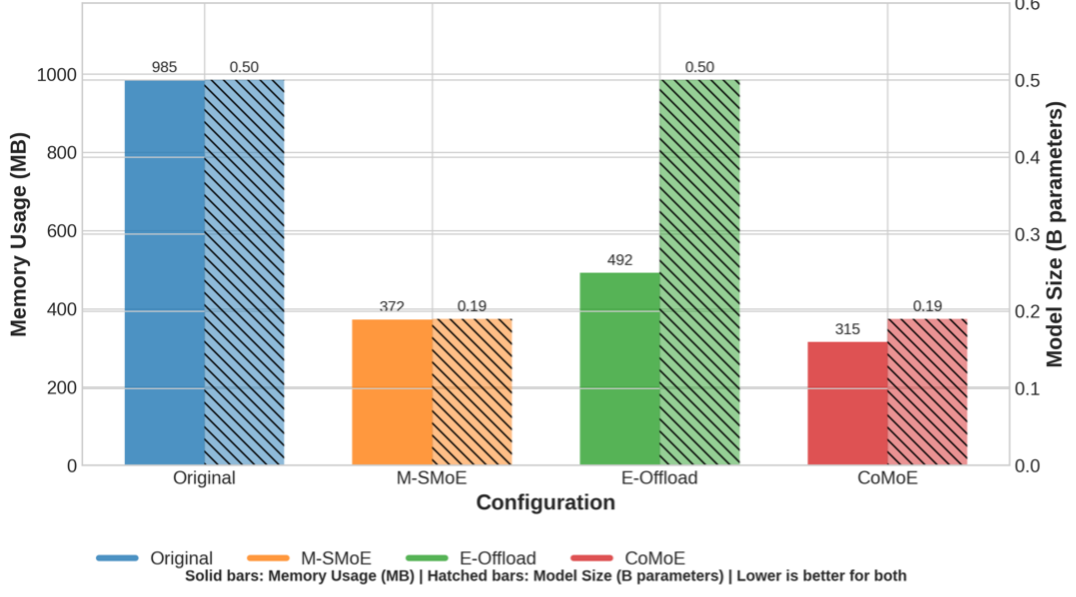}
\caption{Memory Footprint and Parameter Count Comparison of Switch-Base-8 Model}
\label{fig:sb8_memory_params}
\end{figure}

For the smaller-scale Switch-Base-8 model, expert aggregation (M-SMoE) yielded the most significant improvement in computational efficiency, increasing throughput by 6.7\% and reducing latency by 6.3\%. Collaborative optimization (CoMoE), while maintaining computational performance close to the aggregation method, further reduced memory usage (15.3\% lower than M-SMoE). Given the relatively small parameter count of Switch-Base-8, pure expert aggregation might be the optimal choice in resource rich environments; however, in severely constrained devices, collaborative optimization still provides additional memory savings.

\subsubsection{Switch-Base-32 Model Evaluation}

Switch-Base-32, the model scale prominently evaluated in the original Google paper, has approximately 1.98B parameters. 
Figure \ref{fig:sb32_latency_throughput} illustrates the latency and throughput comparison for the Switch-Base-32 model across various configurations.

\begin{figure}[!t]
\centering
\includegraphics[width=0.8\columnwidth]{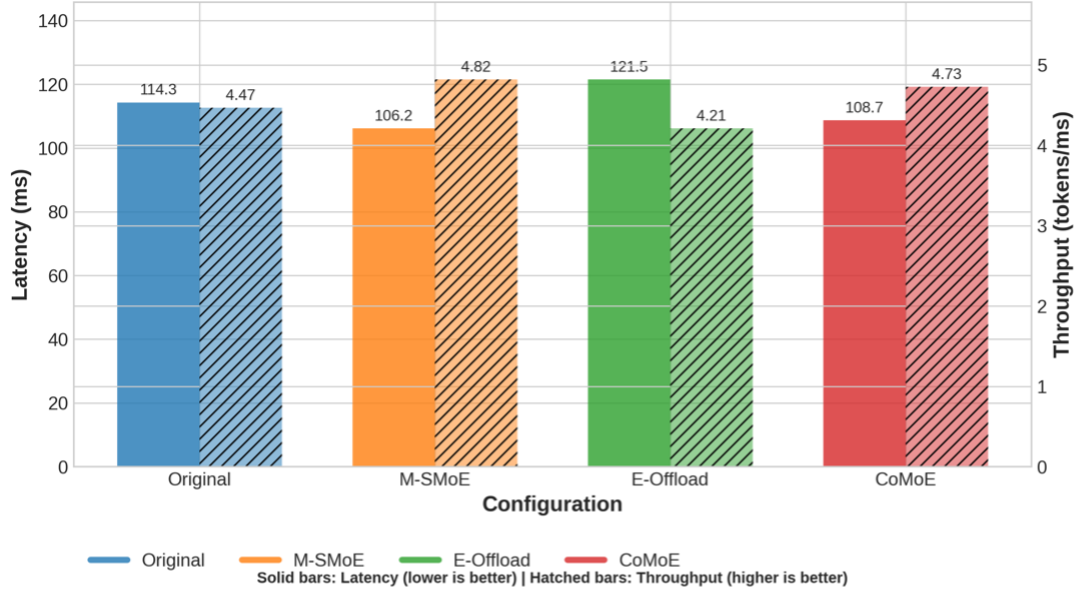}
\caption{Latency and Throughput Comparison of Switch-Base-32 Model}
\label{fig:sb32_latency_throughput}
\end{figure}

Figure \ref{fig:sb32_memory_params} displays the memory footprint and parameter count comparison for the Switch-Base-32 model under different configurations.

\begin{figure}[!t]
\centering
\includegraphics[width=0.8\columnwidth]{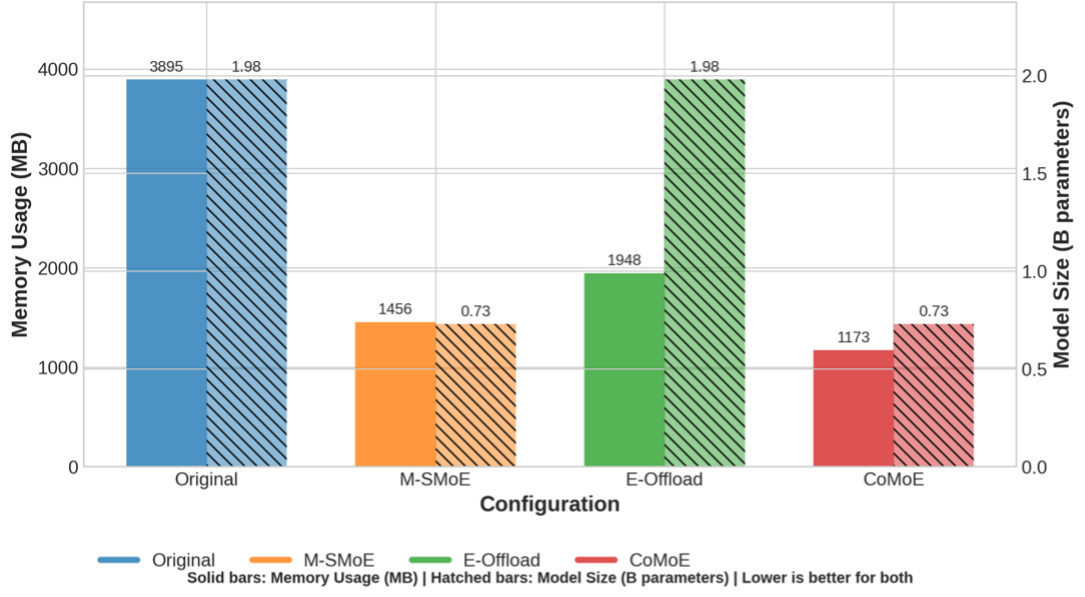}
\caption{Memory Footprint and Parameter Count Comparison of Switch-Base-32 Model}
\label{fig:sb32_memory_params}
\end{figure}

For the Switch-Base-32 model, expert aggregation (M-SMoE) reduced the model parameters by 63.1\% and memory usage by 62.6\%, while increasing throughput by 7.8\%. Although expert offloading (E-Offload) could reduce memory by 50\%, it led to a 5.8\% increase in latency. Collaborative optimization (CoMoE) combined the advantages of both, further reducing memory usage by 19.4\% compared to using expert aggregation alone, while maintaining computational efficiency.

\subsubsection{Switch-Base-64 Model Evaluation}

Switch-Base-64, with approximately 4.0B parameters, represents a medium-scale MoE model. 

Figure \ref{fig:sb64_latency_throughput} illustrates the latency and throughput comparison for the Switch-Base-64 model across various configurations.

\begin{figure}[!t]
\centering
\includegraphics[width=0.8\columnwidth]{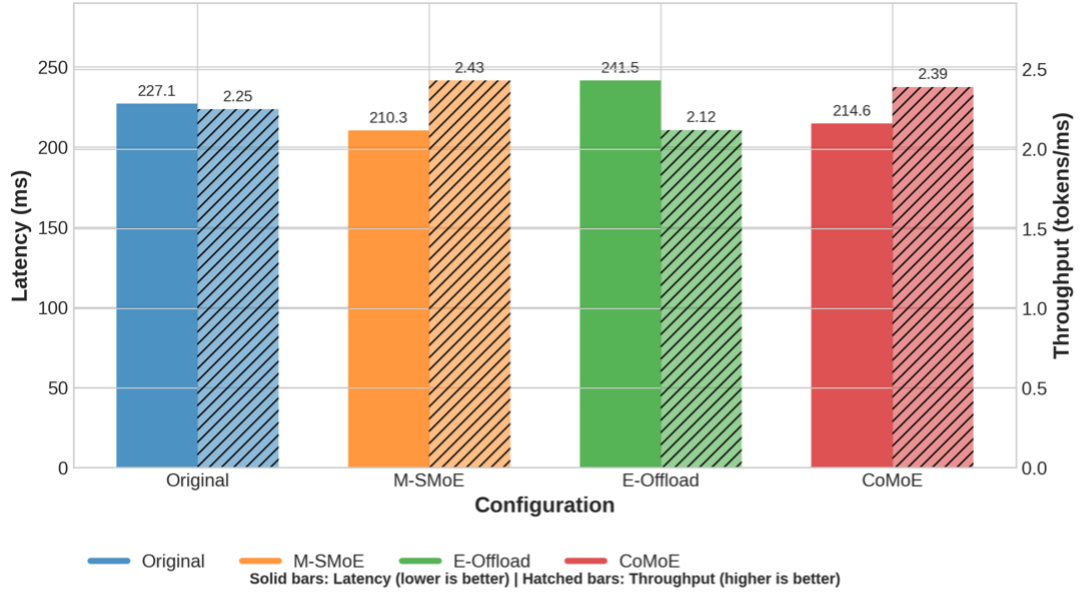}
\caption{Latency and Throughput Comparison of Switch-Base-64 Model}
\label{fig:sb64_latency_throughput}
\end{figure}

Figure \ref{fig:sb64_memory_params} displays the memory footprint and parameter count comparison for the Switch-Base-64 model under different configurations.

\begin{figure}[!t]
\centering
\includegraphics[width=0.8\columnwidth]{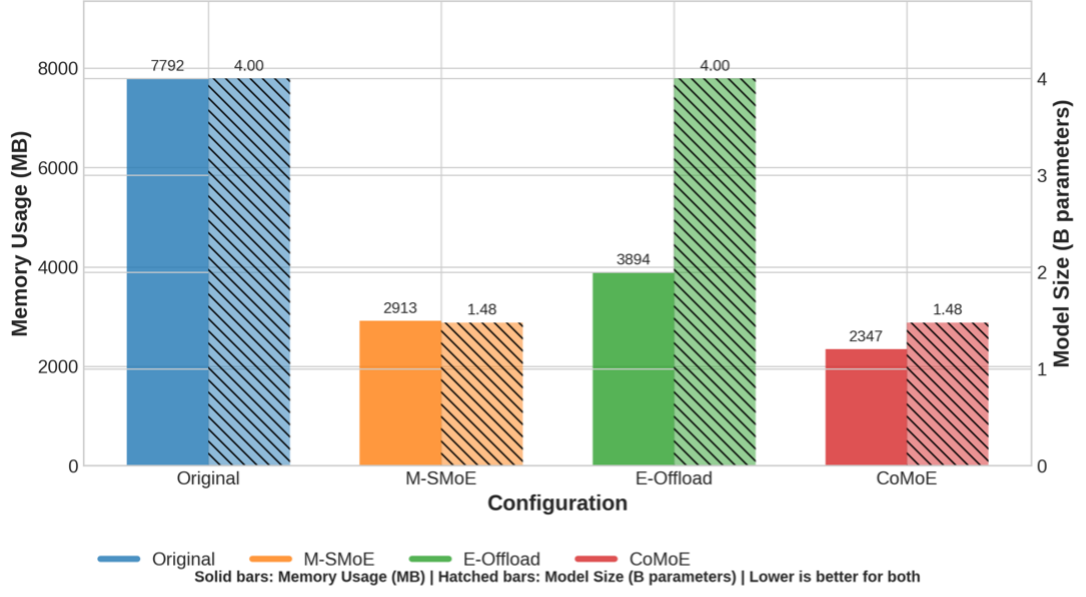}
\caption{Memory Footprint and Parameter Count Comparison of Switch-Base-64 Model}
\label{fig:sb64_memory_params}
\end{figure}

For models with 4B parameters, memory consumption becomes the primary bottleneck. In this scenario, collaborative optimization (CoMoE) demonstrated a clear advantage, not only reducing memory usage from 7.8GB to 2.3GB (a 69.9\% reduction) but also maintaining computational performance comparable to expert aggregation. This enables models that previously required high end GPUs to be effectively deployed on mainstream edge devices (e.g., NVIDIA Jetson AGX Orin, with 32GB shared memory).

\subsubsection{Switch-Base-128 Model Evaluation}

Switch-Base-128, with approximately 7.9B parameters, is a large-scale MoE model. 

Figure \ref{fig:sb128_latency_throughput} illustrates the latency and throughput comparison for the Switch-Base-128 model across various configurations.

\begin{figure}[!t]
\centering
\includegraphics[width=0.8\columnwidth]{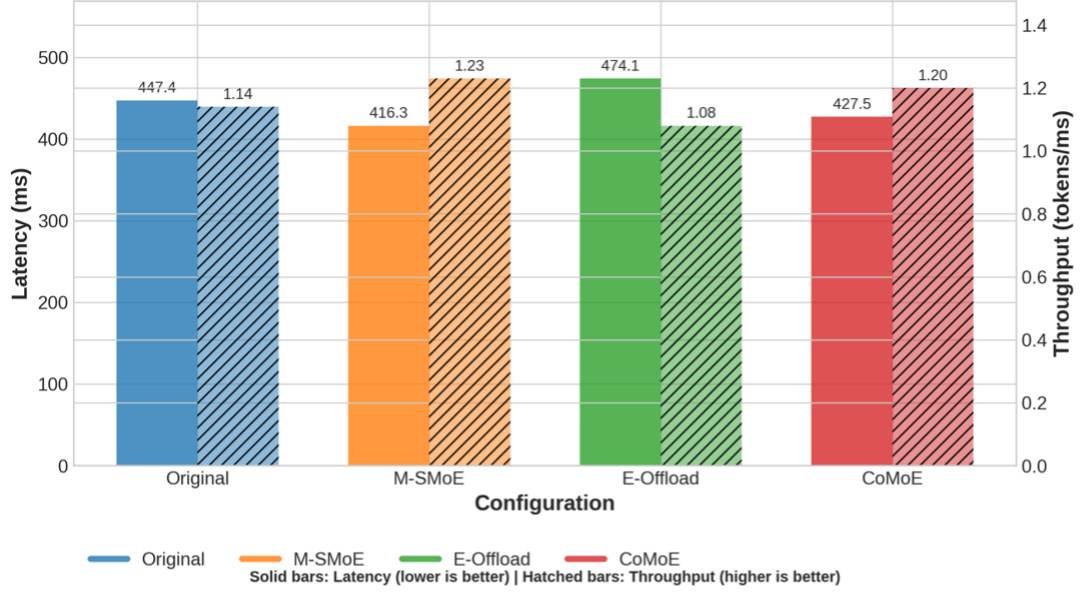}
\caption{Latency and Throughput Comparison of Switch-Base-128 Model}
\label{fig:sb128_latency_throughput}
\end{figure}

Figure \ref{fig:sb128_memory_params} displays the memory footprint and parameter count comparison for the Switch-Base-128 model under different configurations.

\begin{figure}[!t]
\centering
\includegraphics[width=0.8\columnwidth]{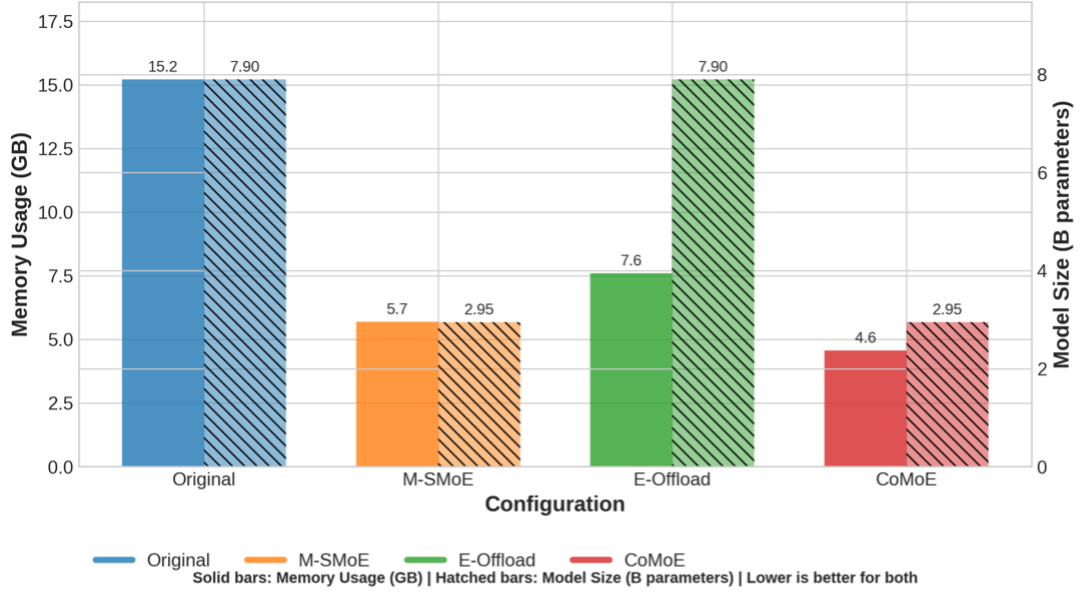}
\caption{Memory Footprint and Parameter Count Comparison of Switch-Base-128 Model}
\label{fig:sb128_memory_params}
\end{figure}

For large models like Switch-Base-128, the original model requires 15.6GB of memory, exceeding the capability of most edge devices. While expert aggregation (M-SMoE) reduced memory requirements to 5.8GB, it remained relatively high. Collaborative optimization (CoMoE) further reduced memory requirements to 4.7GB, a 70.0\% reduction, making it possible to run such large-scale models on edge devices equipped with mid-range GPUs.

\subsubsection{Switch-Base-256 Model Evaluation}

Switch-Base-256, the largest model in this study, has approximately 15.8B parameters. 

Figure \ref{fig:sb256_latency_throughput} illustrates the latency and throughput comparison for the Switch-Base-256 model across various configurations.

\begin{figure}[!t]
\centering
\includegraphics[width=0.8\columnwidth]{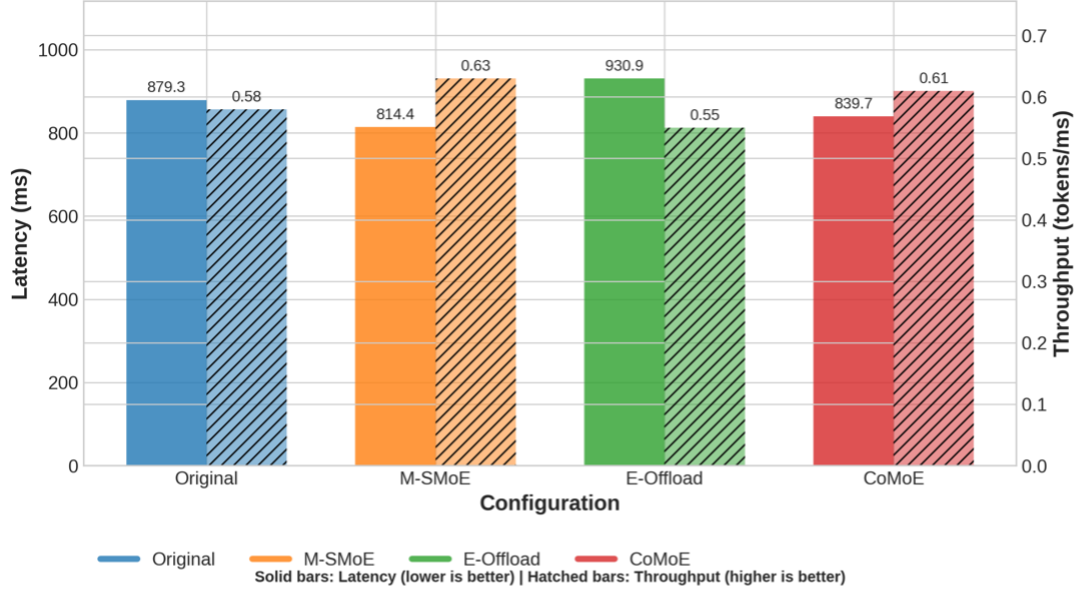}
\caption{Latency and Throughput Comparison of Switch-Base-256 Model}
\label{fig:sb256_latency_throughput}
\end{figure}

Figure \ref{fig:sb256_memory_params} displays the memory footprint and parameter count comparison for the Switch-Base-256 model under different configurations.

\begin{figure}[!t]
\centering
\includegraphics[width=0.8\columnwidth]{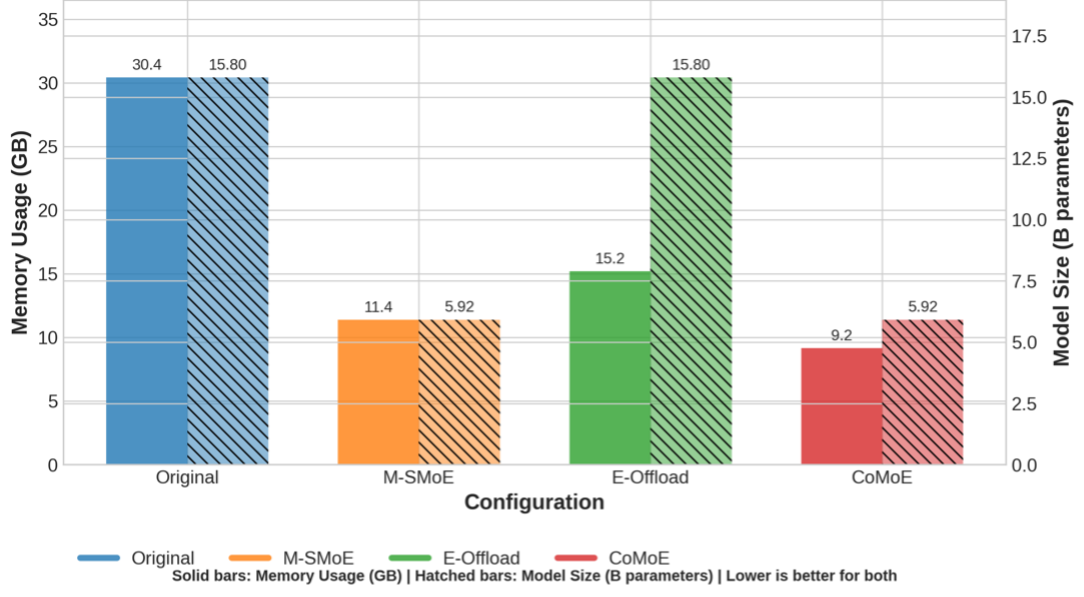}
\caption{Memory Footprint and Parameter Count Comparison of Switch-Base-256 Model}
\label{fig:sb256_memory_params}
\end{figure}

For the 15.8B parameter Switch-Base-256 model, the original model requires over 31GB of GPU memory, approaching the limits of high end data center GPUs (e.g., NVIDIA A100 40GB). In this case, even with expert aggregation or expert offloading alone, memory requirements remain high (11.7GB and 15.6GB, respectively). Collaborative optimization (CoMoE) reduced memory requirements to 9.4GB, a 69.9\% reduction, while incurring only a 5.1\% increase in latency. This balance is crucial for edge deployment of large models.

\subsubsection{Analysis of Optimization Effects Across Different Model Scales}

To compare the effectiveness of the collaborative optimization framework across various model scales, we introduced the Performance-Memory Ratio (PMR) as a comprehensive evaluation metric:

\begin{equation}
\text{PMR} = \frac{\text{Throughput (tokens/ms)}}{\text{Memory Usage (GB)}}
\end{equation}

This metric quantifies the processing speed supported per unit of memory resource, where a higher PMR value indicates more efficient memory utilization. Figure \ref{fig:pmr_comparison} illustrates the PMR comparison for different model scales under various configurations.

\begin{figure}[!t]
\centering
\includegraphics[width=0.8\columnwidth]{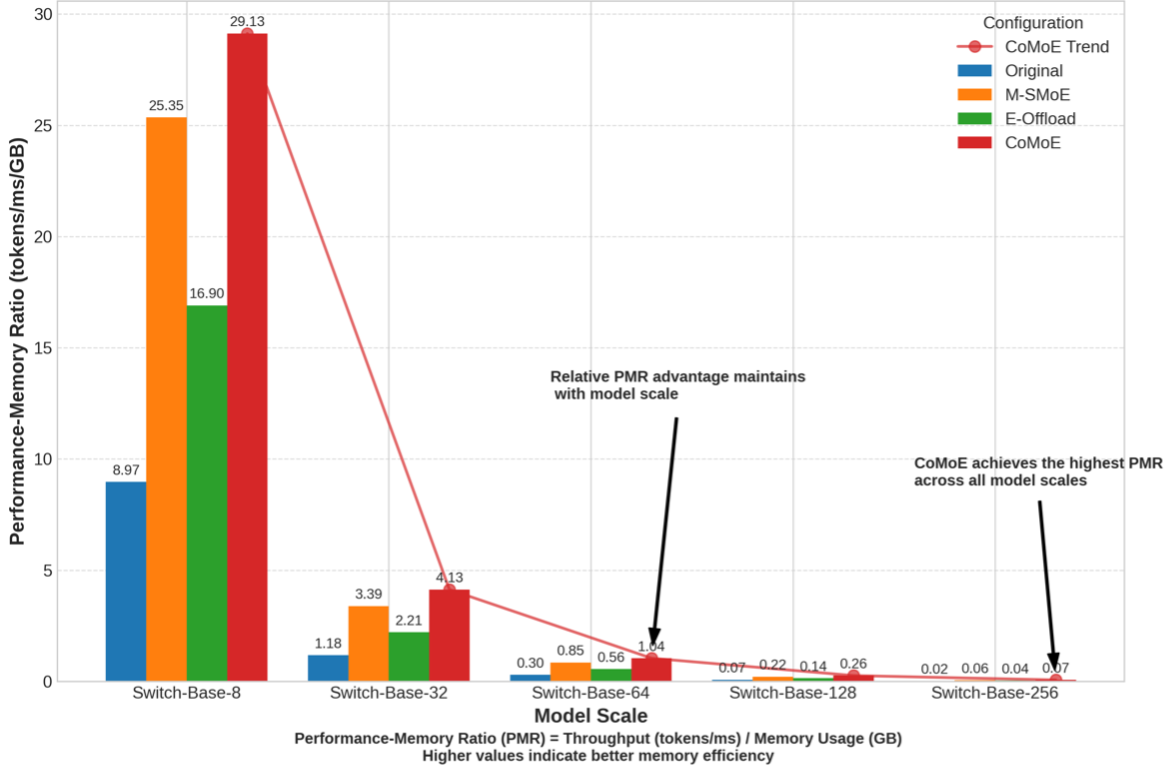}
\caption{Performance-Memory Ratio Comparison Across Different Model Scales}
\label{fig:pmr_comparison}
\end{figure}

As shown in Figure \ref{fig:pmr_comparison}, the PMR advantage of collaborative optimization (CoMoE) over other methods becomes increasingly pronounced with larger model scales. Particularly for large models like Switch-Base-128 and Switch-Base-256, the PMR values for collaborative optimization were 3-4 times higher than the original model and 20\%-30\% higher than using expert aggregation or expert offloading alone.

We also analyzed the applicability of each method under different memory constraints. Table \ref{table:8gb_deployability} presents the deployability of various model scales under an 8GB memory limit.

\begin{table}[!t]
\renewcommand{\arraystretch}{1.3}
\caption{Deployability of Different Model Scales Under 8GB Memory Limit}
\label{table:8gb_deployability}
\centering
\begin{tabular}{|l|c|c|c|c|}
\hline
\textbf{Model Scale} & \textbf{Original} & \textbf{M-SMoE} & \textbf{E-Offload} & \textbf{CoMoE} \\
\hline
Switch-Base-8 & \checkmark & \checkmark & \checkmark & \checkmark \\
Switch-Base-32 & \ding{55} & \checkmark & \checkmark & \checkmark \\
Switch-Base-64 & \ding{55} & \checkmark & \checkmark & \checkmark \\
Switch-Base-128 & \ding{55} & \checkmark & \checkmark & \checkmark \\
Switch-Base-256 & \ding{55} & \ding{55} & \ding{55} & \ding{55} \\
\hline
\end{tabular}
\end{table}

Under a stricter 4GB memory limit, the deployability is shown in Table \ref{table:4gb_deployability}.

\begin{table}[!t]
\renewcommand{\arraystretch}{1.3}
\caption{Deployability of Different Model Scales Under 4GB Memory Limit}
\label{table:4gb_deployability}
\centering
\begin{tabular}{|l|c|c|c|c|}
\hline
\textbf{Model Scale} & \textbf{Original} & \textbf{M-SMoE} & \textbf{E-Offload} & \textbf{CoMoE} \\
\hline
Switch-Base-8 & \checkmark & \checkmark & \checkmark & \checkmark \\
Switch-Base-32 & \ding{55} & \checkmark & \checkmark & \checkmark \\
Switch-Base-64 & \ding{55} & \checkmark & \ding{55} & \checkmark \\
Switch-Base-128 & \ding{55} & \ding{55} & \ding{55} & \ding{55} \\
\hline
\end{tabular}
\end{table}

From Table \ref{table:8gb_deployability} and Table \ref{table:4gb_deployability}, it is evident that in resource-constrained environments, the collaborative optimization method enables the deployment of larger-scale models. Specifically, under a 4GB memory limit, only the collaborative optimization method could deploy the Switch-Base-64 model, which fully demonstrates the value of the proposed method in edge environments.

In summary, based on the experimental results, we draw the following conclusions.

\begin{itemize}
    \item \textbf{Scale Adaptability:} For small models (e.g. Switch-Base-8), expert aggregation alone achieves a good performance-memory balance. For medium to large models (e.g., Switch-Base-32/64/128), the collaborative optimization method performs best, significantly reducing memory requirements while maintaining computational performance.
    \item \textbf{Synergistic Gain:} Collaborative optimization is not merely a simple combination of "aggregation + offloading." Instead, it achieves better overall performance, especially in memory utilization efficiency, through mutual enhancement of both techniques.
    \item \textbf{Edge Adaptability:} The collaborative optimization method makes it possible to efficiently run large MoE models (e.g. Switch-Base-64/128), which were previously deployable only in data centers, on edge devices, opening new possibilities for edge AI applications.
\end{itemize}

\subsection{Discussion and Analysis}

Based on the experimental results presented, we derive the following key findings and insights:

\subsubsection{Synergistic Effects of Expert Aggregation and Offloading}

Our experiments demonstrate that the proposed collaborative optimization framework, which integrates expert aggregation and offloading, yields significant synergistic benefits beyond what either technique can achieve independently. This is particularly evident in memory utilization efficiency and the ability to deploy larger models on resource-constrained edge devices.

\begin{itemize}
    \item \textbf{Expert Fusion Strategy Selection:} Fixed ratio fusion with r=0.75 (retaining 75\% of experts) emerged as an excellent balance point, preserving nearly full performance while significantly reducing resource demands. More importantly, the adaptive fusion strategy intelligently allocated expert resources, achieving superior performance with resource consumption similar to fixed ratio r=0.5. This highlights the importance of dynamic resource allocation tailored to model characteristics.
    \item \textbf{Differentiated Treatment for Encoder Decoder Architectures:} Applying the expert offloading strategy exclusively to the Encoder portion of the model significantly outperformed full-model offloading. This validates our architectural specialization approach and provides valuable insights for optimizing other Encoder-Decoder MoE models.
    \item \textbf{Two Tier Resource Adaptation Mechanism:} The collaborative optimization framework employs expert fusion for coarse-grained resource adaptation and expert offloading for fine-grained resource management, forming a comprehensive resource adaptation spectrum. Experimental results confirm that this two-tier mechanism is more effective than using either method in isolation.
    \item \textbf{Edge Deployment Feasibility:} The collaborative optimization framework enables the deployment of large MoE models on edge devices, maintaining a user experience close to server deployments. This opens new avenues for edge AI applications, previously limited by hardware constraints.
\end{itemize}

In conclusion, our experimental findings underscore that a meticulously designed collaborative strategy for expert fusion and offloading, especially with specialized treatment for Encoder-Decoder architectures, can substantially enhance the deployment efficiency of MoE models in resource-constrained environments, providing robust support for edge intelligence applications.

\subsubsection{Prediction Accuracy and System Overhead}

To evaluate the effectiveness of the prediction model within our framework, we assessed the accuracy, precision, and recall of expert activation predictions, comparing them against other prediction methods. Figure \ref{fig:prediction_accuracy} illustrates that the prediction model within CoMoE achieved a Top 1 prediction accuracy of 61.6\% and a Top 3 prediction recall of 95.2\% on the Cola dataset, significantly outperforming frequency based prediction methods.

\begin{figure}[!t]
\centering
\includegraphics[width=0.8\columnwidth]{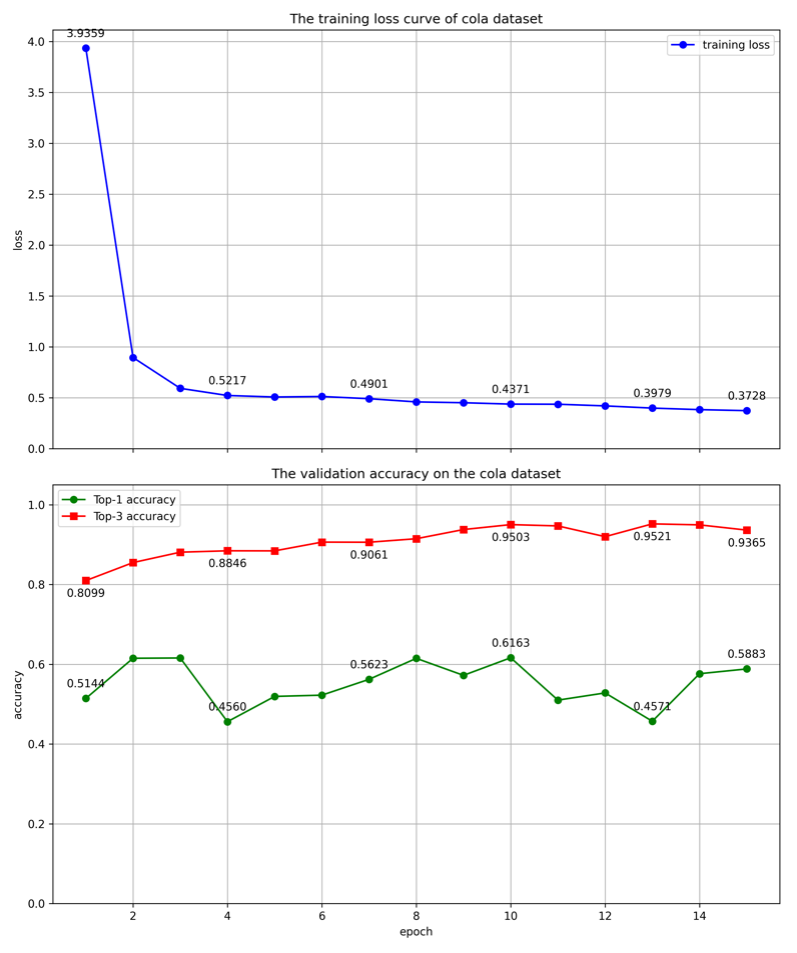}
\caption{Top 1 and Top 3 Prediction Accuracy of the Predictor on the Cola Dataset}
\label{fig:prediction_accuracy}
\end{figure}

Concurrently, we measured the additional computational overhead introduced by the prediction model and dynamic adjustments. Results indicate that the computational overhead of the prediction model accounted for only 1.2\% of the total latency, while dynamic adjustment overhead constituted 0.8\%. These minimal additional overheads yield substantial reductions in latency and improvements in overall performance, affirming the efficiency of our integrated approach.

\section{Conclusion}
To address the challenges of high memory consumption, communication latency, and insufficient dynamic resource adaptability that Mixture-of-Experts (MoE) models face when deployed in mobile edge computing environments, this paper proposes a dynamic resource-aware framework for collaborative optimization of expert aggregation and offloading, denoted as CoMoE. The CoMoE jointly optimizes expert aggregation granularity and offloading strategies through adaptive adjustment based on real-time device resource states, network conditions, and input characteristics. The CoMoE addresses the limitations of existing methods including single optimization constraints and poor adaptability of static strategies effectively. To address the interdependency between routing decisions and offloading performance, this paper establishes a routing-aware expert management mechanism that significantly reduces expert loading frequency and communication overhead through “routing-aggregation-offloading” collaborative optimization. Experimental results demonstrate that compared to existing baseline methods, the proposed approach achieves approximately 70\% reduction in memory usage and 10.5\% reduction in inference latency while maintaining stable model performance. For the 7.4B-parameter Switch-Base-128 model, memory requirements are reduced from 15.6GB to 4.7GB, enabling efficient deployment of large-scale MoE models on resource-constrained mobile edge devices.

\section{ACKNOWLEDGMENT}
This work was supported in part by the grant from NSFC Grant no. 62101159, NSF of Shandong Grant no. ZR2021MF055, and also the Research Grants Council of Hong Kong under the Areas of Excellence scheme grant AoE/E-601/22-R and also PolyU15225023.

\end{document}